\documentclass[11pt]{article}

\usepackage{fullpage}
\usepackage{amsthm,amssymb,amsmath}  
\usepackage{xspace,enumerate}
\usepackage[utf8]{inputenc}
\usepackage{thmtools}
\usepackage{thm-restate}
\usepackage{authblk}

  \theoremstyle{plain}
  \newtheorem{theorem}{Theorem}
  \newtheorem{lemma}[theorem]{Lemma}  
  \newtheorem{corollary}[theorem]{Corollary}  
  \newtheorem{proposition}[theorem]{Proposition}

  \newtheorem{definition}[theorem]{Definition}
  \newtheorem{example}[theorem]{Example}
  \newtheorem*{remark}{Remark}
 
\title{Exact Distance Oracles for Planar Graphs with Failing Vertices\thanks{This work was partially supported by Israel Science Foundation (ISF) grants 794/13 and 592/17.}}

\author{Panagiotis Charalampopoulos}
\author{Shay Mozes}
\author{Benjamin Tebeka}

\affil{
Efi Arazi School of Computer Science, The Interdisciplinary Center Herzliya, Israel\\
\texttt{pcharalampo@gmail.com, smozes@idc.ac.il, benitbk@gmail.com}
}

\date{\vspace{-5ex}}

  \usepackage{hyperref}
  \hypersetup{colorlinks=true,citecolor=blue}
  
\usepackage[margin=1in]{geometry}
\usepackage{graphicx}
\usepackage{microtype}
\usepackage{xspace}
\usepackage{amsmath,amsfonts}
\usepackage{makecell}
\usepackage[ruled,noline,noend]{algorithm2e}
\usepackage{tabularx}
\usepackage{comment}
\usepackage{todonotes}
\usepackage{thmtools}
\usepackage{thm-restate}
\usepackage[capitalise]{cleveref}
\usepackage{verbatim}
\usepackage{units}
\usepackage{color}
\definecolor{darkblue}{rgb}{0,0.08,0.45}
  \usepackage{epsfig}
    \usepackage{graphics}
\usepackage{pgf,tikz}
  \usetikzlibrary{trees, arrows, shapes, snakes}
\usepackage{marvosym}
  \usepackage[caption=false]{subfig}
\usetikzlibrary{decorations.pathmorphing}
\usetikzlibrary{decorations.markings}
\usetikzlibrary{decorations.pathmorphing,shapes}
\usetikzlibrary{calc,decorations.pathmorphing,shapes}

\newcommand{\cO}{\mathcal{O}}
\newcommand{\TG}{\mathcal{T}}
\newcommand{\cOtilde}{\tilde{\cO}}
\newcommand{\Vor}{\textsf{Vor}}
\newcommand{\cone}{\textsf{cone}}

\begin{document}

\maketitle

\thispagestyle{empty}

\begin{abstract}
We consider exact distance oracles for directed weighted planar graphs in the presence of failing vertices. Given a source vertex $u$, a target vertex $v$ and a set $X$ of $k$ failed vertices, such an oracle returns the length of a shortest $u$-to-$v$ path that avoids all vertices in $X$.
We propose oracles that can handle any number $k$ of failures.
We show several tradeoffs between space, query time, and preprocessing time.
In particular, for a directed weighted planar graph with $n$ vertices and any constant $k$, we show an $\cOtilde(n)$-size, $\cOtilde(\sqrt{n})$-query-time oracle.\footnote{The $\cOtilde(\cdot)$ notation hides polylogarithmic factors.} 
We then present a space vs. query time tradeoff: for any $q \in \lbrack 1,\sqrt n \rbrack$, we propose an oracle of size $n^{k+1+o(1)}/q^{2k}$ that answers queries in $\cOtilde(q)$ time.
For single vertex failures ($k=1$), our $n^{2+o(1)}/q^2$-size, $\cOtilde(q)$-query-time oracle improves over the previously best known tradeoff of Baswana et al.~\lbrack SODA 2012\rbrack \  by polynomial factors for $q \geq n^t$, for any $t \in (0,1/2]$. For multiple failures, no planarity exploiting results were previously known.

A preliminary version of this work was presented in SODA 2019. In this version, we show improved space vs. query time tradeoffs relying on the recently proposed almost optimal distance oracles for planar graphs \lbrack Charalampopoulos et al., STOC 2019; Long and Pettie, SODA 2021\rbrack.
\end{abstract}

\clearpage
\setcounter{page}{1}

\section{Introduction}

Computing shortest paths is one of the most well-studied algorithmic problems.
In the data structure version of the problem, the aim is to compactly store information about a graph such that the distance (or the shortest path) between any queried pair of vertices can be retrieved efficiently.
Data structures supporting distance queries are called \emph{distance oracles}.
The two main measures of efficiency of a distance oracle are the space it occupies and the time it requires to answer a distance query.
Another quantity of interest is the time required to construct the oracle.

In recent decades researchers have investigated the shortest path problem in graphs subject to failures, or more broadly, to changes. One such variant is the \emph{replacement paths problem}. In this problem we are given a graph $G$ and vertices $u$ and $v$. The goal is to report the $u$-to-$v$ distance in $G$ for each possible failure of a single edge along the shortest $u$-to-$v$ path. 
Another variant is that of constructing a distance oracle that answers $u$-to-$v$ distance queries subject to edge or vertex failures ($u,v$ and the set of failures are given at query time). Perhaps the most general of these variants is designing a \emph{fully-dynamic} distance oracle; a data structure that supports distance queries as well as updates to the graph such as changes to edge lengths, edge insertions or deletions and vertex insertions or deletions.

One obvious but important application of handling failures is in geographical routing~\cite{DBLP:conf/nsdi/KimGKS05}.
Further motivation for studying this problem originates from Vickrey pricing in networks~\cite{DBLP:conf/stoc/NisanR99,DBLP:conf/focs/HershbergerS01}; see~\cite{DBLP:journals/siamcomp/DemetrescuTCR08} for a concise discussion on the relation between the problems.
A long-studied generalization of the shortest path problem is the $k$-shortest paths problem, in which not one but but several shortest paths must be produced between a pair of vertices. This problem reduces to running $k$ executions of a replacement paths algorithm, and has many applications itself~\cite{DBLP:journals/siamcomp/Eppstein98}. 

In this paper we focus on these problems, and in particular on handling vertex failures in planar graphs. 
Observe that edge failures easily reduce to vertex failures.
Indeed, by replacing each edge $(a,c)$ of $G$ with a new dummy vertex $b$ and appropriately weighted edges $(a,b)$ and $(b,c)$; the failure of edge $(a,c)$ in $G$ corresponds to the failure of vertex $b$ in the new graph. Note that this transformation does not depend on planarity. In sparse graphs, such as planar graphs, this transformation only increases the number of vertices by a constant factor.
Also note that there is no such obvious reduction in the other direction that preserves planarity. 
In general graphs, one can replace each vertex $v$ by two vertices $v_{in}$ and $v_{out}$, assign to $v_{in}$ (resp.~$v_{out}$) all the edges incoming to $v$ (resp.~outgoing from $v$) and add a 0-length directed edge $e$ from $v_{in}$ to $v_{out}$. The failure of vertex $v$ in the original graph corresponds to the failure of edge $e$ in the new graph. However, this transformation does not preserve planarity.

\subsection{Related Work}
\paragraph{General Graphs.} Demetrescu et al.~presented an $\cO(n^2 \log n)$-size oracle answering single failure distance queries in constant time~\cite{DBLP:journals/siamcomp/DemetrescuTCR08}. Bernstein and Karger, improved the construction time in~\cite{DBLP:conf/stoc/BernsteinK09}. Interestingly, Duan and Pettie, building upon this work, showed an $\cO(n^2 \log^3 n)$-size oracle that can report distances subject to two failures, in time $\cO(\log n)$~\cite{DBLP:conf/soda/DuanP09a}.
Based on this oracle, they then easily obtain an $\cOtilde(n^k)$-size oracle answering distance queries in $\cOtilde(1)$ time for any $k \geq 2$.
Oracles that require less space for more than $2$ failures have been proposed, such as the ones presented in~\cite{DBLP:journals/talg/WeimannY13,DBLP:conf/focs/BrandS19}, but at the expense of $\Omega (n)$ query time.
Such oracles are unsatisfactory for planar graphs, where single source shortest paths can be computed in linear or nearly linear time.

\paragraph{Planar Graphs.} 
Exact (failure-free) distance oracles for planar graphs have been studied extensively over the past three decades~\cite{DBLP:conf/wg/Djidjev96,DBLP:conf/esa/ArikatiCCDSZ96,DBLP:conf/stoc/ChenX00,FR,DBLP:conf/soda/MozesS12,DBLP:journals/algorithmica/Cabello12,DBLP:conf/focs/Cohen-AddadDW17,Vorexact,ourexact,DBLP:journals/corr/abs-2009-14716,DBLP:journals/corr/abs-2007-08585}.
A very recent series of papers~\cite{Vorexact,DBLP:conf/focs/Cohen-AddadDW17,ourexact,DBLP:journals/corr/abs-2007-08585,DBLP:conf/esa/Charalampopoulos20} has established Voronoi diagrams as a useful tool for designing distance oracles in planar graphs.
In particular, in~\cite{DBLP:journals/corr/abs-2007-08585}, the authors showed an $n^{1+o(1)}$-size, $\log^{2+o(1)}n$-query-time oracle.

As for handling failures, the replacement paths problem (i.e.~when both the source and destination are fixed in advance) can be solved in nearly linear time~\cite{DBLP:journals/talg/EmekPR10,DBLP:journals/talg/KleinMW10,DBLP:conf/soda/Wulff-Nilsen10}.
For the single source, single failure version of the distance oracle problem (i.e.~when the source vertex is fixed at construction time, and the query specifies just the target and a single failed vertex), Baswana et al.~\cite{DBLP:conf/soda/BaswanaLM12} presented an oracle with size and construction time $\cO(n \log^4 n)$ that answers queries in $\cO(\log^3 n)$ time. 
They then showed an oracle of size $\cOtilde(n^2/q)$ for the general single failure problem (i.e.~when the source, destination, and failed vertex are all specified at query time), that answers queries in time $\cOtilde(q)$ for any $q \in [1,n^{1/2}]$.
They concluded the paper by asking whether it is possible to design a compact distance oracle for a planar digraph which can handle multiple vertex failures. We answer this question in the affirmative.

Fakcharoenphol and Rao, in their seminal paper~\cite{FR}, presented distance oracles that require $\cO(n^{2/3} \log^{7/3} n)$ and $\cO(n^{4/5} \log^{13/5} n)$  amortized time per update and query for non-negative and arbitrary edge-weight updates respectively.\footnote{Though this is not mentioned in~\cite{FR}, the query time can be made worst case rather than amortized by standard techniques.} The space required by these oracles is $\cO(n \log n)$.
Klein presented a similar data structure in~\cite{MSSP} for the case where edge-weight updates are non-negative, requiring time $\cO(n^{2/3} \log^{5/3} n)$.
Klein's result was extended in~\cite{DBLP:conf/stoc/ItalianoNSW11}, where, assuming non-negativity of edge-weight updates, the authors showed how to handle edge deletions and insertions (not violating the planarity of the embedding), and in~\cite{DBLP:journals/talg/KaplanMNS17}, where the authors showed how to handle negative edge-weight updates, all within the same time complexity.
In fact, these results can all be combined, and along with a recent slight improvement on the running time of FR-Dijkstra~\cite{DBLP:conf/icalp/GawrychowskiK18}, they yield a dynamic distance oracle that can handle any of the aforementioned edge updates and queries within time $\cO(n^{2/3} \frac{\log^{5/3} n}{\log^{4/3} \log n})$.
We further extend these results by showing that vertex deletions and insertions can also be handled within the same time complexity. The main challenge lies in handling vertices of high degree.

An exact fault-tolerant distance labeling scheme for planar graphs, accommodating for a single failure was recently presented~\cite{DBLP:conf/sirocco/Bar-NatanCGMW21}.
For the case where one is willing to settle for approximate distances, Abraham et al.~\cite{DBLP:conf/stoc/AbrahamCG12} gave a $(1+\epsilon)$ labeling scheme for undirected planar graphs with polylogarithmic size labels, such that a $(1+\epsilon)$-approximation of the distance between vertices $u$ and $v$ in the presence of $|F|$ vertex or edge failures can be recovered from the labels of $u,v$ and the labels of the failed vertices in $\cOtilde(|F|^2)$ time. They then use this labeling scheme to devise a fully dynamic $(1+\epsilon)$-distance oracle with size $\cOtilde(n)$ and $\cOtilde(\sqrt{n})$ query and update time.\footnote{A fully dynamic distance oracle supports arbitrary edge and vertex insertions and deletions, and length updates.}

On the lower bounds side, it is known that an exact dynamic oracle requiring amortized time $\cO(n^{1/2-\delta})$, for any constant $\delta > 0$, for both edge-weight updates and distance queries, would refute the APSP conjecture, i.e.~that there is no truly subcubic combinatorial algorithm for solving the all-pairs shortest path problems in weighted (general) graphs~\cite{DBLP:conf/focs/AbboudD16}.

\subsection{Our Results and Techniques}
In this work we focus on distance queries subject to vertex failures in planar graphs. Our results can be summarized as follows.

\begin{enumerate}
\item We show how to preprocess a directed weighted planar graph $G$ in $\cOtilde(n)$ time into an oracle of size $\cOtilde(n)$ that, given a source vertex $u$, a target vertex $v$, and a set $X$ of $k$ failed vertices, reports the length of a shortest $u$-to-$v$ path in $G \setminus X$ in  $\cOtilde(\sqrt{kn})$ time. See~\cref{lem:fr}. 
\item We extend the exact dynamic distance oracles mentioned in the previous section to also handle vertex insertions and deletions without changing their space and time bounds. See~\cref{thm:dyn}.
\item For $k$ allowed failures, and for any $r \in [1,n]$, we show how to construct an $n^{k+1+o(1)}/r^{k}$-size oracle that answers queries in time $\cOtilde(k\sqrt{r})$.
For $k=1$, this improves over the previously best known tradeoff of Baswana et al.~\cite{DBLP:conf/soda/BaswanaLM12} by polynomial factors for $r \geq n^t$, for any $t \in (0,1]$.
To the best of our knowledge, this is the first tradeoff for $k>1$. See~\cref{fig:tradeoff} for an illustration and~\cref{cor:main,thm:main2} for more tradeoffs.
\end{enumerate}

\begin{figure*}[ht]
\resizebox{\textwidth}{!}{
\definecolor{uququq}{rgb}{0.25,0.25,0.25}
\definecolor{xdxdff}{rgb}{0.49,0.49,1}
\subfloat{
\begin{tikzpicture}[line cap=round,line join=round,>=triangle 45,x=6cm,y=5cm]
\draw[->,color=black] (-0.025,0) -- (1.1,0);
\draw[color=black] (0.0,2pt) -- (0.0,-2pt) node[below] {\footnotesize $1$};
\draw[color=black] (0.5,2pt) -- (0.5,-2pt) node[below] {\footnotesize $3/2$};
\draw[color=black] (1,2pt) -- (1,-2pt) node[below] {\footnotesize $2$};

\draw[->,color=black] (0,-0.025) -- (0,0.6);
\draw[color=black] (2pt,0) -- (-2pt,0) node[left] {\footnotesize $0$};
\draw[color=black] (2pt,0.5) -- (-2pt,0.5) node[left] {\footnotesize $1/2$};
\draw[color=black] (1,-0.15) node {$\lg S/\lg n$};
\draw[color=black] (0.16,0.55) node {$\lg Q/\lg n$};

\clip(-0.1,-0.1) rectangle (1.1,0.6);

\draw (0.5,0.5)-- (1,0);
\fill [color=uququq] (0.5,0.5) circle (1.5pt);

\draw[dashed] (0,0.5) -- (1,0);

\fill [color=uququq] (1,0) circle (1.5pt);

\draw[color=black] (0.975,0.1) node {\small \cite{DBLP:conf/soda/DuanP09a,DBLP:journals/siamcomp/DemetrescuTCR08,DBLP:conf/stoc/BernsteinK09}};
\draw[color=black] (0.775,0.31) node {\small \cite{DBLP:conf/soda/BaswanaLM12}};
\draw[color=black] (0.35,0.23) node {\small [Sec.~\ref{sec:newtradeoff}]};
\end{tikzpicture}}
\subfloat{
\begin{tikzpicture}[line cap=round,line join=round,>=triangle 45,x=6cm,y=5cm]
\draw[->,color=black] (-0.025,0) -- (1.1,0);
\draw[color=black] (0.0,2pt) -- (0.0,-2pt) node[below] {\footnotesize $1$};
\draw[color=black] (0.25,2pt) -- (0.25,-2pt) node[below] {\footnotesize $2$};
\draw[color=black] (0.5,2pt) -- (0.5,-2pt) node[below] {\footnotesize $3$};
\draw[color=black] (0.75,2pt) -- (0.75,-2pt) node[below] {\footnotesize $4$};
\draw[color=black] (1,2pt) -- (1,-2pt) node[below] {\footnotesize $5$};

\draw[->,color=black] (0,-0.025) -- (0,0.6);
\draw[color=black] (2pt,0) -- (-2pt,0) node[left] {\footnotesize $0$};
\draw[color=black] (2pt,0.5) -- (-2pt,0.5) node[left] {\footnotesize $1/2$};
\draw[color=black] (1,-0.15) node {$\lg S/\lg n$};
\draw[color=black] (0.16,0.55) node {$\lg Q/\lg n$};

\clip(-0.1,-0.1) rectangle (1.1,0.6);

\draw[dashed] (0,0.5) -- (0.25,0);

\draw[dashed, color=red] (0,0.5) -- (0.25,0.25);
\fill [color=red] (0.25,0) circle (1.5pt);

\draw[dashed, color=green] (0,0.5) -- (0.5,0.16667);
\fill [color=green] (0.5,0) circle (1.5pt);

\draw[dashed, color=blue] (0,0.5) -- (0.75,0.125);
\fill [color=blue] (0.75,0) circle (1.5pt);

\draw[dashed, color=magenta] (0,0.5) -- (1,0.1);
\fill [color=magenta] (1,0) circle (1.5pt);

\draw[color=black] (0.6,0.4) node {\small [Sec.~\ref{sec:newtradeoff}] };
\draw[color=black] (0.25,0.055) node {\small \cite{DBLP:conf/soda/DuanP09a}};
\draw[color=black] (0.5,0.055) node {\small \cite{DBLP:conf/soda/DuanP09a}};
\draw[color=black] (0.75,0.055) node {\small \cite{DBLP:conf/soda/DuanP09a}};
\draw[color=black] (1,0.055) node {\small \cite{DBLP:conf/soda/DuanP09a}};
\end{tikzpicture}
}

}
\caption{Left: Tradeoff of the Space ($S$) vs.~the Query time ($Q$) for exact distance oracles for a single failed vertex (i.e.~$k=1$) on a doubly logarithmic scale, hiding subpolynomial factors. The previous tradeoff is indicated by a solid line, while the new tradeoff is indicated by a dashed line.
Right: The same tradeoff for $k=1,\ldots ,5$, shown with different colours.
The points on the $x$-axis correspond to the result of~\cite{DBLP:conf/soda/DuanP09a}, while the new tradeoffs are indicated by dashed lines.\label{fig:tradeoff}}
\end{figure*}

Our nearly-linear space oracle that reports distances in the presence of $k$ failures in $\cOtilde(\sqrt{kn})$ time is obtained by adapting a technique of Fakcharoenphol and Rao~\cite{FR}. 
In a nutshell, a planar graph can be recursively decomposed using small cycle separators, such that, in each level of the decomposition, the boundary of each piece
(i.e.~the vertices of the piece that also belong to other pieces in this level) 
is a union of a constant number of cycles with relatively few vertices. Instead of working with the given planar graph, one computes distances over its \emph{dense distance graph} (DDG); a non-planar graph on the boundary vertices of the pieces which captures the distances between boundary vertices within each of the underlying pieces.  
Fakcharoenphol and Rao developed an efficient implementation of Dijkstra's algorithm on the DDG. This  algorithm, nicknamed \emph{FR-Dijkstra}, runs in time roughly proportional to the number of vertices of the DDG (i.e.~boundary vertices), rather than in time proportional to the number of vertices in the planar graph. Roughly speaking, Fakcharoenphol and Rao show that to obtain distances from $u$ to $v$ with $k$ edge failures, it (roughly) suffices to consider just the boundary vertices of the pieces in the recursive decomposition that contain failed edges. 
Since pieces at the same level of the recursive decomposition are edge-disjoint, 
the total number of boundary vertices in all the required pieces is only $\cO(\sqrt{kn})$. This $\cOtilde(n)$-size, $\cOtilde(\sqrt{kn})$-query-time oracle, supporting distance queries subject to a batch of $k$ edge cost updates, leads to their dynamic distance oracle. 

The difficulty in handling vertex failures is that a high degree vertex $x$ may be a boundary vertex of many (possibly $\Omega(n)$) pieces in the recursive decomposition. Then, if $x$ fails, one would have to consider too many pieces and too many boundary vertices. Standard techniques such as degree reduction by vertex splitting are inappropriate because when a vertex fails all its copies fail. To overcome this difficulty we define a variant of the dense distance graph which, instead of capturing shortest path distances between boundary vertices within a piece, only captures distances of paths that are internally disjoint from the boundary. We show that such distances can be computed efficiently, and that it then suffices to include in the FR-Dijkstra computation (roughly) only pieces that contain $x$, but not as a boundary vertex. This leads to our nearly-linear-size oracle reporting distances in the presence of $k$ failures in $\cOtilde(\sqrt{kn})$ time (item 1 above). See~\cref{sec:FR}.
Plugging the same technique into the existing dynamic distance oracles extends them to support vertex deletions (item 2 above). See~\cref{sec:dyn}.

Our main result, the space vs. query time tradeoff (item 3 above), is obtained by a combination of this technique, employment of external $DDG$s, and the recent static exact distance oracle presented in~\cite{DBLP:journals/corr/abs-2007-08585}. See~\cref{sec:newtradeoff}. 
In the case where one is willing to sacrifice space in order to make preprocessing more efficient, we show an alternative tradeoff in~\cref{sec:oldtradeoff}.
Such an oracle could be preferable in the case that one has to reconstruct the data structure every once in a while due to unfixable failures or other updates in the graph.
This tradeoff is achieved by a combination of FR-Dijkstra on our variant of the DDG with $r$-divisions, external $DDG$s, and efficient point location in Voronoi diagrams ---a tool that is used internally by the exact oracles we use as a black box in the other trafeoff.
Finally, in~\cref{sec:prec} we show how to efficiently construct our oracles; in particular, the efficient construction of external $DDG$s may be of independent interest.

\section{Preliminaries}\label{sec:prel}

In this section we review the main techniques required for describing our result. Throughout the paper we consider a weighted directed planar graph $G=(V(G),E(G))$, embedded in the plane.
(We use the terms weight and length for edges and paths interchangeably throughout the paper.)
We use $|G|$ to denote the number of vertices in $G$. Since planar graphs are sparse, $|E(G)| = \cO(|G|)$ as well.
For an edge $(u,v)$, we say that $u$ is its tail and $v$ is its head.
$d_G(u,v)$ denotes the distance from $u$ to $v$ in $G$.
We denote by $d_G(u,v,X)$ the distance from $u$ to $v$ in $G\setminus X$, where $X \in V(G)$ or $X \subset V(G)$; if the reference graph is clear from the context we may omit the subscript.
We assume that the input graph has no negative length cycles.
If it does, we can detect this in $\cO(n \frac{\log^2 n}{\log \log n})$ time by computing single source shortest paths from any vertex~\cite{DBLP:conf/esa/MozesW10}.
In the same time complexity, we can transform the graph in a standard way so that all edge weights are non-negative and shortest paths are preserved.
We further assume that shortest paths are unique as required for a result from~\cite{Vordiam} that we use; this can be ensured in $\cO(n)$ time by a deterministic perturbation of the edge weights~\cite{DBLP:conf/stoc/0001FL18}.
Each original distance can be recovered from the corresponding distance in the transformed graph in constant time.

\paragraph{Separators and recursive decompositions in planar graphs.}
Miller~\cite{DBLP:conf/stoc/Miller84} showed how to compute a Jordan curve that intersects the graph at $\cO(\sqrt{n})$ vertices and separates it into two pieces with at most $2n/3$ vertices each. Jordan curve separators can be used to recursively separate a planar graph until pieces have constant size.
The authors of~\cite{DBLP:conf/stoc/KleinMS13} show how to obtain a complete recursive decomposition tree $\TG$ of $G$ in $\cO(n)$ time. 
$\TG$ is a binary tree whose nodes correspond to subgraphs of $G$ (pieces), with the root being all of $G$ and the leaves being pieces of constant size.
For each vertex $u$ of $G$, we fix an arbitrary leaf-piece in $\TG$ that contains $u$, and denote this piece by $P_u$.
We identify each piece $P$ with the node representing it in $\TG$. We can thus abuse notation and write $P\in \TG$.

An $r$-division~\cite{DBLP:journals/siamcomp/Frederickson87} of a planar graph, for $r \in [1,n]$, is a decomposition of the graph into $\cO(n/r)$ pieces, each of size $\cO(r)$, such that each piece has $\cO(\sqrt{r})$ boundary vertices, i.e.~vertices incident to edges in other pieces.
Another usually desired property of an $r$-division is that the boundary vertices lie on a constant number of faces of the piece (holes).
For every $r$ larger than some constant, an $r$-division with this property (i.e.~few holes per piece) is represented in the decomposition tree $\TG$ of~\cite{DBLP:conf/stoc/KleinMS13}.
Throughout the paper, to avoid confusion, we use ``nodes" when referring to $\TG$ and ``vertices" when referring to $G$.
We denote the boundary vertices of a piece $P$ by $\partial P$. We refer to non-boundary vertices as internal.

\begin{lemma}[\cite{Vorexact}]\label{lem:rdiv}
Each node in $\TG$ corresponds to a piece such that
(i) each piece has $\cO(1)$ holes,
(ii) the number of vertices in a piece at depth $\ell$ in $\TG$ is $\cO(n/c^\ell_1)$, for some constant $c_1 > 1$,
(iii) the number of boundary vertices in a piece at depth $\ell$ in $\TG$ is $\cO(\sqrt{n}/c^\ell_2)$, for some constant $c_2 > 1$.
\end{lemma}

We use the following well-known bounds (see e.g.,~\cite{Vorexact}).

\begin{proposition}\label{prop:1}
$\sum_{P \in \TG}|P|=\cO(n \log n)$, $\sum_{P \in \TG}|\partial P|=\cO(n)$ and $\sum_{P \in \TG}|\partial P|^2=\cO(n \log n)$.
\end{proposition}

We show the following bound that will be used in future proofs.

\begin{proposition}\label{prop:2}
$\sum \limits_{P \in \TG} |P| |\partial P|^2 = \cO(n^2)$.
\end{proposition}
\begin{proof}
Let $P^\ell_1,P^\ell_2,\ldots , P^\ell_j$ be the pieces at the $\ell$-th level of the decomposition.
$\sum_i |P^\ell_i|=\cO(n)$ since the pieces are edge-disjoint. We know by~\cref{lem:rdiv} that $|\partial P^\ell_j| = \cO(\sqrt{n}/c^\ell_2)$ for all $j$ and hence $|\partial P^\ell_j|^2 = \cO(n/c^{2\ell}_2)$ for all $j$.
It follows that $\sum_i |P^\ell_i| |\partial P^\ell_i|^2 = \cO(n^2/c^{2\ell}_2)$ and the claimed bound follows by summing over all levels of $\TG$.
\end{proof}

\paragraph{Dense distance graphs and FR-Dijkstra.}
The \emph{dense distance graph} of a piece $P$, denoted 
$DDG_P$ is a complete directed graph on the boundary vertices of $P$.
Each edge $(u,v)$ has weight $d_{P}(u,v)$, equal to the length of the shortest $u$-to-$v$ path in $P$.
$DDG_P$ can be computed in time $\cO((|\partial P|^2 + |P|) \log |P|)$ using the multiple source shortest paths (MSSP) algorithm~\cite{MSSP,DBLP:journals/siamcomp/CabelloCE13}.
Over all pieces of the recursive decomposition this takes time $\cO(n \log^2 n)$ in total and requires space $\cO(n \log n)$ by~\cref{prop:1}.
We next give a ---convenient for our purposes--- interface for FR-Dijkstra~\cite{FR}, which is an efficient implementation of Dijkstra's algorithm on any union of $DDG$s.
The algorithm exploits the fact that, due to planarity, certain submatrices of the adjacency matrix of $DDG_P$ satisfy the Monge property.
(A matrix $M$ satisfies the Monge property if, for all $i<i'$ and $j<j'$, $M_{i,j}+M_{i',j'} \leq M_{i',j}+M_{i,j'}$~\cite{monge1781memoire}.) The interface is specified in the following theorem, which was essentially proved in~\cite{FR}, with some additional components and details from~\cite{DBLP:journals/talg/KaplanMNS17,DBLP:conf/esa/MozesW10}.

\begin{theorem}[\cite{FR,DBLP:journals/talg/KaplanMNS17,DBLP:conf/esa/MozesW10}]\label{thm:FR}
A set of $DDG$s with $\cO(M)$ vertices in total (with multiplicities), each having at most $m$ vertices, can be preprocessed in time and extra space $\cO(M \log m)$ in total, so that, after this preprocessing, Dijkstra's algorithm can be run on the union of any subset of these $DDG$s with $\cO(N)$ vertices in total (with multiplicities) in time $\cO(N \log N \log m)$, by relaxing edges in batches. Each such batch consists of edges that have the same tail.
\end{theorem}

The algorithm in the above theorem is called FR-Dijkstra. It is useful in computing distances in sublinear time, as demonstrated by~\cref{lem:frcone} and~\cref{cor:fr} which are a reformulation of ideas from~\cite{FR} and are provided below for completeness.

\begin{definition}
The \emph{cone} of a vertex $u$ of $G$ is the union of the following DDGs: (i) $DDG_{P_u}$, with $u$ considered a boundary vertex of $P_u$. (ii) For every (not necessarily strict) ancestor $P$ of $P_u$, $DDG_Q$ of the sibling $Q$ of $P$.
\end{definition}

\begin{lemma}\label{lem:frcone}
Let $x$ and $y$ be two vertices in the cone of a vertex $u$. The $x$-to-$y$ distance in $G$ equals the $x$-to-$y$ distance in this cone of $u$.
\end{lemma}
\begin{proof}
Let $P_u = P_0, P_1, \dots, P_d = G$ be the ancestors of $P_u$ ordered by decreasing depth in $\TG$.
Let $Q_i$ be the sibling of $P_i$ in $\TG$.
Let $\cone_i$ be $DDG_{P_u} \cup \bigcup_{j<i} DDG_{Q_j}$.
We will prove by induction that for any two vertices $x,y \in \cone_i$, the $x$-to-$y$ distance in $P_i$ equals the $x$-to-$y$ distance in $\cone_i$.
This statement is trivially true for $i=0$.
Let us assume it is true for $k$.
Consider an $x$-to-$y$ shortest path $p$ in $P_{k+1}$, where $x,y \in \cone_{k+1}$. 
Path $p$ can be decomposed into maximal subpaths that are entirely contained in $P_k$ or $Q_k$ and whose endpoints are in $\{x,y\} \cup (\partial P_k \cap \partial Q_k)$. For each such subpath we either have a path with the same length in $\cone_k$ by the inductive assumption,  or an edge of $DDG_{Q_k}$. This shows that the length of $p$ is at least the length of the $x$-to-$y$ distance in $\cone_k$. Since every edge of $\cone_k$ corresponds to some path in $P_k$, the opposite also holds, so the two quantities are equal.
\end{proof}

\begin{corollary}\label{cor:fr}
Let $u,v$ be two distinct vertices in $G$. Let $p$ be a shortest $u$-to-$v$ path in $G$. If $p$ is not fully contained in $P_u$ then we can compute the length of $p$ by running FR-Dijkstra on the union of the cone of $u$ and the cone of $v$. This takes $\cOtilde(\sqrt n)$ time.
\end{corollary}
\begin{proof}
Since $p$ is not fully contained in $P_u$, $p$ must visit a vertex $w$ in the separator of the LCA of $P_u$ and $P_v$ in $\TG$. We are done by decomposing $p$ into the prefix ending at $w$ and the suffix beginning at $w$, and applying Lemma~\ref{lem:frcone}. The running time follows by Theorem~\ref{thm:FR} and Lemma~\ref{lem:rdiv}.
\end{proof}

\section{Near Linear Space Data Structure for any Number of Failures}\label{sec:FR}

In this section we show how to adapt the approach of~\cite{FR} for distance oracles supporting cumulative edge changes to support distance queries with failed vertices.
The main technical challenge lies in dealing with failures of high-degree vertices, since such vertices may belong to many pieces at each level of the decomposition. For example, think of a failure of the central vertex in a wheel graph, which belongs to all the pieces in the recursive decomposition. Note that standard degree reduction techniques such as vertex splitting are not useful because when a vertex fails all its copies fail.
This is in contrast with  the situation when dealing only with edge-weight updates, since each edge can be in at most one piece per level.
We circumvent this by defining and employing the {\em strictly internal dense distance graph}.
The main intuition is that strictly internal DDGs enable us to handle pieces that only contain failed boundary vertices, i.e.~do not contain any internal vertex that fails. Then, only pieces that contain internal failed vertices are ``problematic". Note however, that a vertex is internal in at most one piece per level of the decomposition.

\begin{definition}
The \emph{strictly internal} dense distance graph of a piece $P$, denoted $DDG^\circ_P$, is a complete directed graph on the boundary vertices of $P$.
An edge $(u,v)$ has weight $d^\circ_{P}(u,v)$ equal to the length of the shortest $u$-to-$v$ path in $P$ that is internally disjoint from $\partial P$.
\end{definition}

The sole difference to the standard definition is that in our case paths are not allowed to go through $\partial P$.
Observe that the shortest path in $P$ between two vertices of $\partial P$ is still represented in $DDG^\circ_P$, just not necessarily by a single edge as in $DDG_P$.
This establishes the following lemma.
\begin{lemma}
For any piece $P$ and any two boundary vertices $u,v \in \partial P$, the $u$-to-$v$ distance in $DDG^\circ_P$ equals the $u$-to-$v$ distance in $DDG_P$.
\end{lemma}

We now discuss how to efficiently compute $DDG^\circ_P$.
We construct a planar graph $\hat{P}$, by creating a copy of $P$ and incrementing the weight of each edge $uv$, such that $u \in \partial P$, by $C=2\sum_{e \in E(G)}|w(e)|$.
$DDG_{\hat{P}}$ can be computed in $\cO((|\partial P|^2 + |P|) \log |P|)$ time using MSSP~\cite{MSSP,DBLP:journals/siamcomp/CabelloCE13}.
Observe that any $u$-to-$v$ path in $\hat{P}$ that starts at $\partial \hat P$ and is internally disjoint from $\partial \hat{P}$ has exactly one edge $uw$ with $u \in \partial P$, so its length is at least $C$ and less than $2C$, while any $u$-to-$v$ path that has an internal vertex in $\partial P$ is of length at least $2C$.
Therefore, the $u$-to-$v$ distance in $\hat P$ is equal to $C$ plus the length of the shortest $u$-to-$v$ path in $P$ that is internally disjoint from $\partial P$ if the latter one is not $\infty$.
We thus set $d^\circ_{P}(u,v)=d_{\hat{P}}(u,v)-C$.
This completes the description of the computation of $DDG^\circ_P$.
Note that since $C$ is defined in terms of $G$ rather than $P$, edge weights greater than $C$ in $DDG^\circ_P$ effectively represent infinite length in the sense that such edges will never be used by any shortest path (in $P$ nor in $G$). 
Also note that it follows directly from the definition of the Monge property that subtracting $C$ from each entry of a Monge matrix preserves the Monge property. Therefore, we can use $\bigcup_P DDG^\circ_P$ in FR-Dijkstra (Theorem~\ref{thm:FR}) instead of $\bigcup_P DDG_P$.

\paragraph{Preprocessing.} We compute a complete recursive decomposition tree $\TG$ of $G$ in time $\cO(n)$ as discussed in~\cref{sec:prel}. 
We compute $DDG^\circ_P$ for each non-leaf piece $P \in \TG$ and preprocess it as in FR-Dijkstra.
By~\cref{prop:1}, Theorem~\ref{thm:FR} and the above discussion, the time and space complexities are $\cO(n \log^2 n)$ and $\cO(n \log n)$, respectively.

\paragraph{Query.} Upon query $(u,v,X)$, we run FR-Dijkstra on the union of the following $DDG^\circ$s, which we denote by $\mathcal{D}(u,v,X)$ or just $\mathcal{D}$ when the arguments are clear from the context (inspect~\cref{fig:siblings2} for an illustration):

\begin{figure*}[t]
\centering
\resizebox{\textwidth}{!}{ 
\begin{tikzpicture}

\node[shape=circle, draw=white, label = {right: $P_u$}] (0) at (-4.3,-6) {};
\node[shape=circle, draw=white, label = {left: $P_v$}] (0) at (-7.7,-6) {};
\node[shape=circle, draw=white, label = {left: $P_x$}] (0) at (4.3,-6) {};

\fill [gray!40]  (-7,-6) ellipse (0.75cm and 0.5cm);
\fill [gray!40]  (-5,-6) ellipse (0.75cm and 0.5cm);
\fill [gray!40] (-1.5,-4) ellipse (1.25cm and 0.5cm);
\fill [gray!40] (-6,0) ellipse (2cm and 0.7cm);

\fill [red] (2.5,-2) ellipse (1.75cm and 0.6cm);

\fill [gray!85]  (7,-6) ellipse (0.75cm and 0.5cm);
\fill [gray!85]  (5,-6) ellipse (0.75cm and 0.5cm);
\fill [gray!85] (1.5,-4) ellipse (1.25cm and 0.5cm);
\fill [gray!85] (-2.5,-2) ellipse (1.75cm and 0.6cm);

\draw (-3,2) ellipse (2.3cm and 0.8cm);
\draw (-6,0) ellipse (2cm and 0.7cm);
  \draw[-latex, thick, shorten >= 0.2cm] (-2.5,1.2) -- (-1,0.43);
  \draw[-latex, thick, shorten >= 0.2cm] (-3.5,1.2) -- (-5,0.43);
\node[shape=circle,draw=white, label = {right: $G$}] (0) at (-0.7,2) {};
\node[shape=circle,draw=white, label = {right: $v$}] (0) at (-2.6,2.3) {};
\node[shape=circle,draw=white, label = {right: $u$}] (0) at (-3,2) {};
\node[shape=circle,draw=white, label = {right: $x$}] (0) at (-2.5,1.6) {};

\node[shape=circle,draw=white, label = {right: $v$}] (0) at (0.5,0.4) {};
\node[shape=circle,draw=white, label = {right: $u$}] (0) at (-0.8,0) {};
\node[shape=circle,draw=white, label = {right: $x$}] (0) at (0.3,-0.3) {};

\draw (0,0) ellipse (2cm and 0.7cm);

\draw[color=black] (-2.5,-2) ellipse (1.75cm and 0.6cm);
  \draw[-latex, thick, shorten >= 0.2cm] (-0.45,-0.7) -- (-2,-1.55);
\node[shape=circle,draw=gray!85, label = {right: $u$}] (0) at (-3.3,-2.2) {};
\node[shape=circle,draw=gray!85, label = {right: $v$}] (0) at (-2.7,-2) {};

\draw[color=black] (2.5,-2) ellipse (1.75cm and 0.6cm);
  \draw[-latex, thick, shorten >= 0.2cm] (0.45,-0.7) -- (2,-1.55);
  \node[shape=circle,draw=red, label = {right: $x$}] (0) at (2.5,-2) {};

\draw[color=black] (-4.5,-4) ellipse (1.25cm and 0.5cm);
  \draw[-latex, thick, shorten >= 0.2cm] (-2.85,-2.6) -- (-4.3,-3.63);
\node[shape=circle, scale = 0.2, draw=black, fill, label = {right: $u$}] (0) at (-5.75,-4) {};
  \node[shape=circle,draw=white, label = {right: $v$}] (0) at (-5,-3.8) {};

\draw[color=black] (-1.5,-4) ellipse (1.25cm and 0.5cm);
  \draw[-latex, thick, shorten >= 0.2cm] (-2.15,-2.6) -- (-2.15,-3.75);
    \node[shape=circle, scale = 0.2, draw=black, fill, label = {right: $u$}] (0) at (-2.75,-4) {};

\draw (1.5,-4) ellipse (1.25cm and 0.5cm);
  \draw[-latex, thick, shorten >= 0.2cm] (2.15,-2.6) -- (2.15,-3.75);
      \node[shape=circle, scale = 0.2, draw=black, fill, label = {left: $x$}] (0) at (2.75,-4) {};

\draw[color=black] (4.5,-4) ellipse (1.25cm and 0.5cm);
  \draw[-latex, thick, shorten >= 0.2cm] (2.85,-2.6) -- (4.3,-3.63);
      \node[shape=circle, scale = 0.2, draw=black, fill, label = {right: $x$}] (0) at (3.25,-4) {};

\draw[color=black] (-7,-6) ellipse (0.75cm and 0.5cm);
  \draw[-latex, thick, shorten >= 0.2cm] (-4.5,-4.5) -- (-6.8,-5.7);
\node[shape=circle, scale = 0.2, draw=black, fill, label = {right: $v$}] (0) at (-7.75,-6) {};

\draw[color=black] (-5,-6) ellipse (0.75cm and 0.5cm);
  \draw[-latex, thick, shorten >= 0.2cm] (-4.5,-4.5) -- (-5,-5.7);
        \node[shape=circle, scale = 0.2, draw=black, fill, label = {left: $u$}] (0) at (-4.25,-6) {};
     \node[shape=circle, scale = 0.2, draw=black, fill, label = {right: $v$}] (0) at (-5.75,-6) {};

\draw[color=black] (5,-6) ellipse (0.75cm and 0.5cm);
  \draw[-latex, thick, shorten >= 0.2cm] (4.5,-4.5) -- (5,-5.7);
      \node[shape=circle, scale = 0.2, draw=black, fill, label = {right: $x$}] (0) at (4.25,-6) {};

\draw (7,-6) ellipse (0.75cm and 0.5cm);
  \draw[-latex, thick, shorten >= 0.2cm] (4.5,-4.5) -- (6.8,-5.7);
     \node[shape=circle, scale = 0.2, draw=black, fill, label = {right: $x$}] (0) at (6.25,-6) {};

\end{tikzpicture}
}
\caption{A portion of the complete recursive decomposition tree $\TG$ of a graph $G$.
The light gray and red pieces are the ones that would be considered by the failure-free distance oracle upon query $d(u,v)$.
However, given the failure of vertex $x$, the $DDG^\circ$ of the red piece is invalid.
The dark gray pieces are the ones that our algorithm considers instead of the red piece.
The $DDG^\circ$s of the dark gray pieces, that are descendants of the red piece, allow us to represent its $DDG^\circ$ subject to the failure of $x$.}
\label{fig:siblings2}
\end{figure*}

\begin{enumerate}
\item For each $w \in \{u,v\}$,  $DDG^\circ_{P_w}$ of $P_w \setminus X$ with $w$ regarded as a boundary vertex. This can be computed on the fly in constant time since the size of the leaf-piece $P_w$ is constant.
\item For each $w \in \{u,v\}$, for each ancestor $P$ of $P_w$ (including $P_w$), $DDG^\circ_Q$ of the sibling $Q$ of $P$ if $Q$ does not contain any internal (i.e.~non-boundary) vertex that is in $X$.
\item For each $x \in X$, $DDG^\circ_{P_x}$ of $P_x \setminus X$. This can be computed on the fly in constant time since the size of the leaf-piece $P_x$ is constant.
\item For each $x \in X$, for each ancestor $P$ of $P_x$ (including $P_x$), $DDG^\circ_Q$ of the sibling $Q$ of $P$ if $Q$ does not contain any internal vertex that is in $X$.
\end{enumerate}

We can identify these $DDG^\circ$s in $\cO(k \log n)$ time by traversing the parent pointers from each $P_i$, for $i \in X$, and marking all the nodes that have an internal failed vertex.
We make one small but crucial change to FR-Dijkstra.
When running FR-Dijkstra, we do not relax edges whose tail is a failed vertex.
This guarantees that, although failed vertices might appear in the graph on which FR-Dijkstra is invoked, the $u$-to-$v$ shortest path computed by FR-Dijkstra does not contain any failed vertices. We therefore obtain the following lemma.

\begin{lemma}\label{lem:fr}
There exists a data structure of size $\cO(n \log n)$, which can be constructed in $\cO(n \log^2 n)$ time, and answers the following queries in $\cO(\sqrt{kn}\log^2 n)$ time. Given vertices $u$ and $v$, and a set $X$ of $k$ failed vertices, report the length of a shortest $u$-to-$v$ path that avoids the vertices of $X$. 
\end{lemma}
\begin{proof}
We have already discussed the space occupied by the oracle and the time required to build it. It remains to analyze the query algorithm.

\emph{Correctness.} First, it is easy to see that no edge $(y,z)$ of any of the $DDG^
\circ$s in $\mathcal D$ represents a path containing a vertex $x \in X$, unless $\{y,z\} \cap X \neq \emptyset$.
The latter case does not affect the correctness of the algorithm, since in FR-Dijkstra we do not relax edges whose tail is a failed vertex.
Hence, the algorithm never computes a distance corresponding to a path going through a failed vertex.

It remains to show that the shortest path in $G\setminus X$ is represented in $\mathcal D$.
For this, by Corollary~\ref{cor:fr}, it suffices to prove that for each piece $A$ in the cone of $u$ (and similarly in the cone of $v$), either $DDG^\circ_A$ for $A\setminus X$ belongs to $\mathcal D$, or $\mathcal D$ contains enough information to reconstruct $DDG^\circ_A$ for $A \setminus X$ (i.e. subject to the failures) during FR-Dijkstra. In the latter case we say that $DDG^\circ_A$ is {\em represented} in $\mathcal D$.
Note that, for any piece $P$, $DDG^\circ_P$ is represented in $\mathcal D$ if  the $DDG^\circ$s of its two children in $\TG$ are represented in $\mathcal D$. (This follows by an argument identical to the one used in the proof of Lemma~\ref{lem:frcone}.)
If $A$ contains no internal failed vertex then $DDG^\circ_A$ is in $\mathcal D$ by point $1$ or $2$ above.
We next consider the case that $A$ does contain some failed vertex $x \in X$ as an internal vertex. Thus $A$ is an ancestor of $P_x$. To show that $A$ is represented in $\mathcal D$, we prove that for any failed vertex $y\in X$, the $DDG^\circ$ of any non-root ancestor of $P_y$ in $\TG$ is represented in $\mathcal D$.

We proceed by the minimal counterexample method. For any $y \in X$, $DDG^\circ_{P_y}$ is in $\mathcal D$ since it is computed on the fly in point $3$. Let $F$ be the deepest node in $\TG$ that is a strict ancestor of $P_y$ for some $y \in X$ and whose $DDG^\circ$ is not represented in $\mathcal D$.
It follows that one of $F$'s children must also be an ancestor of $P_y$ and by the choice of $F$ its $DDG^\circ$ is represented in $\mathcal D$. Let the other child of $F$ be $J$.
If $J$ is an ancestor of some $P_z$, $z \in X$, then $DDG_J^\circ$ is also represented in $\mathcal D$ by the choice of $F$.
Otherwise, $J$ does not contain any internal failed vertex, and hence $DDG^\circ_J$ is in $\mathcal D$ by point $4$.
In either case, the $DDG^\circ$s of both children of $F$ are represented in $\mathcal D$, so $DDG^\circ_F$ is also represented in $\mathcal D$, a contradiction.

\emph{Time complexity.}
Let $r=n/k$ and consider an $r$-division of $G$ in $\TG$.
The pieces of this $r$-division have $\cO(\frac{n}{\sqrt{r}})=\cO(\sqrt{kn})$ boundary vertices in total and this is known to also be an upper bound on the total number of boundary vertices (with multiplicities) of ancestors of pieces in this $r$-division (cf.~the discussion after Corollary~5.1 in~\cite{Vorexact}).

Recall that we have chosen a leaf-piece $P_i$ for each vertex $i \in \{u,v\}\cup X$.
Each piece (other than the $P_i$s) whose $DDG^\circ$ belongs to $\mathcal D$ is a sibling of an ancestor of some $P_i$.
This implies that each $i \in \{u,v\}\cup X$ contributes the $DDG^\circ$s of at most two pieces per level of the decomposition.
Let the ancestor of $P_i$ that is in the $r$-division be $R_i$.
For each $P_i$, we only need to bound the total size of pieces it contributes that are descendants of $R_i$, since we have already bounded the total size of the rest.
We do so by applying~\cref{lem:rdiv} for the subtree of $\TG$ rooted at each $R_i$.
(The extra $\cO(\sqrt{r})$ boundary vertices we start with do not alter the analysis of this lemma as these many are anyway introduced by the first separation of $R_i$.)
It yields $2 \sum_{\ell} \frac{\sqrt{r}}{c^\ell_2}$, where $c_2>1$, which is $\cO(\sqrt{r})$.
Summing over all $k+2$ pieces $P_i$ we obtain the upper bound $\cO(k\sqrt{r})=\cO(\sqrt{kn})$.

FR-Dijkstra runs in time proportional to the total number of vertices of the $DDG^\circ$s in $\mathcal D$ up to a $\log^2 n$ multiplicative factor and hence the time complexity follows.
\end{proof}

\begin{remark}
By using existing techniques (cf.~\cite[Section~5.4]{DBLP:journals/talg/KaplanMNS17}), we can report the actual shortest path $\rho$ in time $\cO(|\rho| \log \log \Delta_\rho)$, where $\Delta_\rho$ is the maximum degree of a vertex of $\rho$ in $G$.\footnote{This remark also applies to the dynamic distance oracle presented in~\cref{sec:dyn} and to the oracles presented in~\cref{sec:newtradeoff}. However, it does not apply to the oracles presented in~\cref{sec:oldtradeoff}, where we use some $DDG$s without storing MSSP data structures or exact distance oracles for the underlying graphs, which would allow us to retrieve the path underlying each $DDG$ edge efficiently.}
\end{remark}

\section{Dynamic Distance Oracles can Handle Vertex Deletions}\label{sec:dyn}

In this section we briefly explain how the techniques of Section~\ref{sec:FR}, and specifically our notion of strict dense distance graphs ($DDG^\circ$s) can be used to facilitate vertex deletions in dynamic distance oracles for planar graphs.
The dynamic distance oracle of~\cite{FR} for non-negative edge-weight updates was improved and simplified in~\cite{MSSP}.
In~\cite{MSSP}, the algorithm obtains an $r$-division of $G$, and then computes and preprocesses the $DDG$s of the pieces of the $r$-division in $\cO(n \log n)$ time to allow for FR-Dijkstra computations in the union of these $DDG$s in time $\cO(\frac{n}{\sqrt{r}} \log^2 n)$.
For a given query asking for the distance from some vertex $u$ to some vertex $v$, the algorithm performs standard Dijkstra computations within the piece containing $u$ (resp.~$v$) to compute the distances from $u$ to the boundary vertices of the piece (resp.~from the boundary vertices of the piece to $v$).
The algorithm then combines this with an FR-Dijkstra computation on the boundary vertices of the $r$-division.
Given an edge update, only the $DDG$ of the unique piece in the $r$-division containing the updated edge needs to get updated, and this requires $\cO(r \log r)$ time.
The balance is at $r=n^{2/3} \log^{2/3} n$, yielding $\cO(n^{2/3} \log^{5/3} n)$ time per update and query.
This result was extended in~\cite{DBLP:conf/stoc/ItalianoNSW11}, where the authors showed how to allow for edge insertions (not violating the planarity of the embedding) and edge deletions and further in~\cite{DBLP:journals/talg/KaplanMNS17} where the authors showed how to handle arbitrary (i.e.~also negative) edge-weight updates.
The time complexity was improved by a $\log^{4/3} \log n$ factor in~\cite{DBLP:conf/icalp/GawrychowskiK18}.

We observe that, by using $DDG^\circ$s instead of standard $DDG$s, vertex deletions can also be handled as follows.
Each vertex is either a boundary vertex in each piece of the $r$-division containing it, or an internal vertex in a unique piece.
If a deleted vertex is a boundary vertex, we just mark it as such and do not relax edges outgoing from it during (FR-)Dijkstra computations. 
If a deleted vertex is internal, we recompute the $DDG^\circ$ of the piece containing it, and reprocess it in time $\cO(r \log r)$ exactly as in the case of edge-weight updates.
The only slightly technical issue we need to take into account is that in~\cref{sec:FR}, edge weights in $DDG^\circ$ are shifted by the large constant $C$ (recall that $C$ is defined as twice the sum of edge weights in the entire graph $G$). The problem is that $C$ might change after each update operation, and this update affects the weights of all the edges in all $DDG^\circ$s. This can be easily solved using indirection. Instead of using the explicit value of $C$ in each edge weight, we represent $C$ symbolically, and store the actual value of $C$ explicitly at some placeholder.
Updating $C$ can be done in constant time because only the explicit value at the placeholder needs to be updated. Whenever an edge weight is required by the algorithm, it is computed on the fly in constant time using the value of $C$ stored in the placeholder.
The data structures underlying FR-Dijkstra do not make use of any 
integer data structures like predecessor data structures ---all used data structures are comparison based. Hence, since the value of $C$ is greater than all edge-weights at the time they are built, they are identical to the data structures that would have been built for this piece with any subsequent value of $C$.
Vertex additions do not alter shortest paths, and hence can be treated trivially.
Note that, as in~\cite{DBLP:conf/stoc/ItalianoNSW11}, we can afford to recompute the entire data structure from scratch after every $\cO(\sqrt{r})$ operations. This guarantees that the number of vertices and number of boundary vertices in each piece remain  $\cO(r)$ and $\cO(\sqrt r)$, respectively, throughout.
We formalize the above discussion in the following theorem.

\begin{theorem}\label{thm:dyn}
A planar graph $G$ can be preprocessed in time $\cO(n \frac{\log^2 n}{\log \log n})$ so that edge-weight updates, edge insertions not violating the planarity of the embedding, edge deletions, vertex insertions and deletions, and distance queries can be performed in time $\cO(n^{2/3} \frac{\log^{5/3} n}{\log^{4/3} \log n})$ each, using $\cO(n)$ space.
\end{theorem}

\section{Tradeoff I: Space vs. Query Time}\label{sec:newtradeoff}

In this section we describe a tradeoff between the size of the oracle and the query time.

\subsection{More Preliminaries and Notation}
We will be using exact distance oracles as a black-box, and as there are different tradeoffs, we denote the space, the query time and the preprocessing time for such a distance oracle over a planar graph of size $n$ as $\sigma(n)$, $\mu(n)$ and $\tau(n)$, respectively.
Let us formally state results from~\cite{DBLP:journals/corr/abs-2007-08585}.

\begin{theorem}[\cite{DBLP:journals/corr/abs-2007-08585}]\label{thm:exact}
Given a planar graph $G$ of size $n$, there exists a distance oracle that can be built in $n^{3/2+o(1)}$ time and admits either of:
\begin{enumerate}[(a)]
\item $n^{1+o(1)}$ space and $\log^{2+o(1)} n$ query time, or
\item $n \log^{2+o(1)} n$ space and $n^{o(1)}$ query time.
\end{enumerate}
\end{theorem}

We now define another useful modification of dense distance graphs.

\begin{definition}\label{def:ddgminus}
The strictly external dense distance graph $DDG^\circ_{ext}(P_1,\ldots,P_i)$ of $G$ for pieces $P_1,\ldots,P_i$ is a complete directed graph on the boundary vertices of $P_1,\ldots,P_i$. The edge $(u,v)$ has weight equal to the length of the shortest $u$-to-$v$ path in $G\setminus \big(\big(\bigcup\limits_{j=1}^{i} P_j\big)\setminus \{u,v\}\big)$.
\end{definition}

$DDG^\circ_{ext}$s can be preprocessed using Theorem~\ref{thm:FR} together with $DDG^\circ$s so that we can perform efficient Dijkstra computations in any union of $DDG^\circ_{ext}$s and $DDG^\circ$s. 

The number of pieces in an $r$-division is at most $cn/r$ for some constant $c$. For convenience, we define
$$g(n,r,k)=\genfrac(){0pt}{0}{cn/r}{k} \leq \frac{(cn)^k}{r^k k!} \leq \frac{n^k}{r^k k^3},$$

\noindent where the last inequality holds for $k$ larger than some constant depending on $c$. We use $g(n,r,k)$ throughout to encapsulate the  dependency on $k$.

\subsection{The Oracle}

\paragraph{Warm up.} Let us first sketch a warm-up $\cOtilde(\frac{n^3}{r^2})$-size oracle with $\cOtilde(\sqrt{r})$ query time that can handle single failures, using the approach of~\cref{sec:FR}. Suppose that we store $DDG^\circ_{ext}$s for all triples of pieces of the $r$-division and that we have preprocessed these for efficient FR-Dijkstra computations together with the $DDG^\circ$s. The total space required is $\cOtilde(g(n,r, 3)(\sqrt{r})^2)=\cOtilde(\frac{n^3}{r^2})$.
Upon query, we first retrieve pieces $R_u, R_v$ and $R_x$ containing $u$, $v$ and $x$, respectively ---assume for now that these pieces are distinct.
Then, we run FR-Dijkstra on $DDG^\circ_{ext}(R_u,R_x,R_v)$, the cone of $u$ in $R_u$, the cone of $v$ in $R_v$ and the pieces that allow us to represent the $DDG$ of $R_x$ subject to the failure of $x$. The query time is $\cOtilde(\sqrt{r})$.
This approach can be generalized to give an oracle that can handle $k$ failures, by considering $(k+2)$-tuples of pieces of the $r$-division.\footnote{We consider the elements of tuples to be unordered throughout.} The space required is $\cO(n\log n + g(n,r,k+2)((k\sqrt{r})^2+k\sqrt{r} \log n))=\cOtilde(\frac{n^{k+2}}{r^{k+1}})$ and the query time is $\cOtilde(k\sqrt{r})$.

\paragraph{Strategy.} Instead of storing information for $(k+2)$-tuples of pieces as in the warm-up, we will store the analogous information for $(k+1)$-tuples and more information for $k$-tuples.
Given $u,v,X$, where $X=\{x_1 , \ldots x_k\}$, we show how to compute $d_G(u,v,X)$ relying on the information stored for the tuples $(R_u , R_{x_1}, \ldots , R_{x_k})$ and $(R_{x_1}, \ldots , R_{x_k})$.
Let us define $Y=\big(\bigcup_{t \in \{u\} \cup X} \partial R_t \big) \setminus X$.
Our aim is to decompose the sought path on the last vertex of $Y$ it visits.

\paragraph{Auxiliary data structure.}
We define \textsc{$u$-to-$k$-Boundary}$(u,X,\mathcal{R},v)$ queries as follows.
The input is
\begin{itemize}
\item a vertex $u$, 
\item a set of vertices $X=\{x_1, \ldots , x_k\}$ of cardinality $k$,
\item a set $\mathcal{R}$ consisting of at most $k+1$ pieces of a specified $r$-division, such that each $j \in \{u\} \cup X$ is in some $R \in \mathcal{R}$,
\item a vertex $v\in R$ for some $R \in \mathcal{R}$, which may be \textsf{null}.
\end{itemize}
The output of the query is the distance from $u$ to each of the vertices of $\{v\} \cup \big( \bigcup_{R \in \mathcal{R}} \partial R \big) \setminus X$ in $G \setminus X$.\footnote{If $v$ is \textsf{null}, we can just set it to be any vertex in $\big( \bigcup_{R \in \mathcal{R}} \partial R \big) \setminus X$. In what follows, we thus do not treat this case separately.}

\begin{lemma}\label{lem:utob}
There exists a data structure of size $\cO(\frac{n^{k+1}}{r^{k}}\log n)$, which can be constructed in time $\cO(\frac{n^{k+1}}{r^{k}}\log^2 n)$, and answers \textsc{$u$-to-$k$-Boundary}$(u,X,\mathcal{R},v)$ queries in $\cO(k\sqrt{r}\log^2 n)$ time.
\end{lemma}
\begin{proof}
We first perform the precomputations of~\cref{sec:FR}. We also obtain an $r$-division of $G$ from $\TG$ in $\cO(n)$ time. Let us denote the pieces of this $r$-division by $R_1,\ldots ,R_{q}$.

We compute $DDG^\circ_{ext}(R_{i_1}, \ldots , R_{i_{k+1}})$ for each ($k+1$)-tuple $(R_{i_1}, \ldots , R_{i_{k+1}})$ of pieces in the $r$-division.
The $DDG^\circ_{ext}$s for all $(k+1)$-tuples can be computed in time $\cO(\frac{(cn)^{k+1}}{r^k}\frac{1}{(k-1)!} \log^2 n)$ for some constant $c$, as shown in~\cref{lem:prec:ddgext} in~\cref{sec:prec}; this dominates the preprocessing time. 
We preprocess these $DDG^\circ_{ext}$s together with the standard $DDG^\circ$s, using~\cref{thm:FR}, to allow for efficient FR-Dijkstra computations.
The total space required is $\cO(g(n,r,k+1)((k\sqrt{r})^2+k\sqrt{r} \log n))=\cO(\frac{n^{k+1}}{r^{k}}\log n)$.

Let us now consider a \textsc{$u$-to-$k$-Boundary}$(u,X,\mathcal{R},v)$ query.
Let $\mathcal{D}=\mathcal{D}(u,v,X)$ be the set of $DDG^\circ$s specified in the query procedure of~\cref{sec:FR} (and illustrated in~\cref{fig:siblings2}).
Now, let the restriction of $\mathcal{D}$ to $\mathcal{R}$ be defined as \[\mathcal{D}_\mathcal{R}=\{DDG^\circ_P \in \mathcal{D} : P \text{ is a weak descendant of some }R \in \mathcal{R}\}.\]
We then run FR-Dijkstra on the union of $\mathcal{D}_{\mathcal{R}}$ and the $DDG^\circ_{ext}$ of a $(k+1)$-tuple that contains all elements of $\mathcal{R}$, not relaxing edges whose tail is in $X$ if encountered. This takes time $\cO(k\sqrt{r} \log^2 n)$.
\end{proof}

\begin{example}
In order to develop some intuition of the above proof, consider~\cref{fig:siblings2}, and suppose that red piece (say $R_1$) and the unique dark gray piece that contains both $u$ and $v$ (say $R_2$) are pieces of the $r$-division. Then, upon a \textsc{$u$-to-$k$-Boundary}$(u,\{x\},\{R_1,R_2\},v)$ query, we would run FR-Dijkstra on the union of the $DDG^\circ$s of all gray pieces that are weak descendants of either $R_1$ or $R_2$ (corresponding to $\mathcal{D}_\mathcal{R}$) and $DDG^\circ_{ext}(R_1,R_2)$, not relaxing edges whose tail is in $X$ if encountered.
\end{example}

\paragraph{Extra preprocessing: exact distance oracles.}
For each $k$-tuple of pieces $\mathcal{R}=(R_{i_1},\ldots,R_{i_{k}})$ of the $r$-division we compute and store an exact distance oracle for graph $G_\mathcal{R}$, which is obtained from $G\setminus (\cup_{j \in \{i_1, \ldots , i_k\}} R_j \setminus \cup_{j \in \{i_1, \ldots , i_k\}} \partial R_j)$ by increasing the weight of edges whose tail is in $\partial R_j$ for some $j \in \{i_1, \ldots , i_k\}$ by a constant $W$ so that they are not internal vertices in any shortest path. 

We are now ready to describe the query procedure; \cref{fig:fast_query} provides an illustration of the setting.

\begin{figure}[htpb!]
    \centering
    \includegraphics[width=13cm]{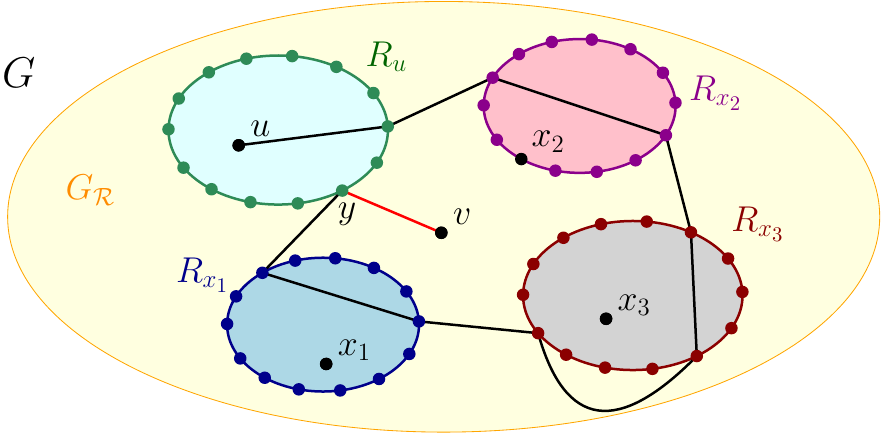}
    \caption{An illustration of the setting of the query.
    The case where $v$ does not lie in any piece from $\mathcal{R}$ is shown.
    Suppose that the shortest $u$-to-$v$ path is as shown. Part I of the query computes the length of its black portion using FR-Dijkstra, while part II computes the length of its red portion using an exact distance oracle for $G_\mathcal{R}$.}
    \label{fig:fast_query}
\end{figure}

\paragraph{Query part I: $u$-to-$Y$.}
With the above lemma at hand, we can easily compute the $u$-to-$Y$ distances.
We first retrieve $k+1$ not necessarily distinct pieces, $R_u, R_{x_1}, \ldots ,R_{x_k}$, such that $u \in R_u$ and $x_i \in R_{x_i}$.
To support that, each vertex stores a pointer to some piece of the $r$-division that contains it.
In the degenerate case that these $v$ is in one of those pieces, then we are done by performing a \textsc{$u$-to-$k$-Boundary}$(u,X,\mathcal{R},v)$ query, with $\mathcal{R}=\{R_u, R_{x_1}, \ldots ,R_{x_k}\}$.
(In order to be able to check whether a vertex is in some particular piece of $\TG$ efficiently, we store, for each piece in $\TG$, a binary tree with the vertices in the piece.)
If on the other hand $v$ does not lie in any piece from $\mathcal{R}$, then any shortest $u$-to-$v$ path must go through some vertex $y \in Y$.
A \textsc{$u$-to-$k$-Boundary}$(u,X,\mathcal{R},\textsf{null})$ query returns the distance from $u$ to each $y \in Y$ in $G \setminus X$.
This takes time $\cO(k\sqrt{r} \log^2 n)$.

\paragraph{Query part II: $Y$-to-$v$.}

For each $y \in Y$ we perform a $y$-to-$v$ distance query in the exact distance oracle for $G_\mathcal{R}$.
We then simply have to take the minimum of $d_G(u,y,X)+d_{G_\mathcal{R}}(y,v)$ among all $y \in Y$ and substract $W$ from this value to retrieve the sought distance.
The time required is $\cO(k\sqrt{r} (\log^2 n+\mu(n)))$.

\begin{theorem}\label{thm:main}
Assume that given a planar graph $G$ with $n$ vertices, one can construct in $\tau(n)$ time a $\sigma(n)$-size distance oracle that answers queries in $\mu(n)$ time.
Then, for any integer $r \in [1,n]$ and for any integer $k \leq \frac{n}{r}$, there exists a data structure of size $\cOtilde(\frac{n^{k}}{r^{k}} \sigma(n))$, which can be constructed in time $\cOtilde(\frac{n^{k}}{r^{k}} \tau(n))$, and can answer the following queries in $\cO(k\sqrt{r} (\log^2 n+\mu(n)))$ time. Given vertices $u$ and $v$ and a set $X$ of at most $k$ failed vertices, report the length of a shortest $u$-to-$v$ path in $G \setminus X$.
\end{theorem}

\begin{remark}
Our distance oracle can actually handle any number $f$ of failures that lie in at most $k$ pieces of the $r$-division in time $\cOtilde((k+\sqrt{fk})\sqrt{r})$. This follows from the fact that the $DDG^\circ$s we will add for a piece with $f_i$ failures have total size $\cOtilde(\sqrt{f_i r})$ by the same analysis as in the proof of~\cref{lem:fr} and the fact that, given $f_1,\ldots,f_k$ such that $\sum_{i=1}^k f_i=f$, we have $\sum_{i=1}^k \sqrt{f_i}\leq \sqrt{fk}$ by the Cauchy-Schwarz inequality. This remark applies to~\cref{thm:main2} as well.
\end{remark}

\begin{proof}[Proof of Theorem~\ref{thm:main}.]
The time complexity of the query algorithm is analyzed above. We next analyze its correctness, the space required by our data structure and its construction time.

\emph{Query correctness.}
Let $\rho$ be a shortest $u$-to-$v$ path in $G \setminus X$.
Let $z$ be the last vertex of $\rho$ that belongs to $Y$, if any such vertex exists.
The distance $d_G(u,z,X)$ from $u$ to $z$ in $G\setminus \{x\}$ is computed by the FR-Dijkstra, while the distance from $z$ to $v$ in $G\setminus X$ is obtained from the query to the exact distance oracle.
In the complementary case, in which no vertex of $\rho$ is in $Y$, we have that $v \in R_u$ and hence the sought distance is computed by the \textsc{$u$-to-$k$-Boundary}$(u,X,\mathcal{R},v)$ query.
It is easy to see that we do not obtain any distance that does not correspond to an actual path in $G\setminus X$.

\emph{Space complexity.}
The space required for the exact distance oracles is $\cOtilde(\frac{n^{k}}{r^{k}} \sigma(n))$. This dominates the $\cO(\frac{n^{k+1}}{r^{k}} \log n)$ space required by the data structure of~\cref{lem:utob} in the $\cOtilde(\cdot)$ notation.

\emph{Preprocessing time.} 
This is also dominated by the $\cO(\frac{n^{k}}{r^{k}} \tau(n))$ time required to build the exact distance oracles, at least in the $\cOtilde(\cdot)$ notation.
\end{proof}

By combining~\cref{thm:exact,thm:main} we obtain the following tradeoffs.

\begin{corollary}\label{cor:main}
Given a planar graph $G$ of size $n$, there exists a distance oracle that supports up to $k$ vertex failures that can be built in time $\frac{n^{k}}{r^{k}} \cdot n^{3/2+o(1)}$ and admits either of:
\begin{enumerate}[(a)]
\item $\frac{n^{k}}{r^{k}} \cdot n^{1+o(1)}$ space and $k\sqrt{r} \cdot \log^{2+o(1)} n$ query time, or
\item $\cOtilde\left(\frac{n^{k}}{r^{k}} \cdot n\right)$ space and $k\sqrt{r} \cdot n^{o(1)}$ query time.
\end{enumerate}
\end{corollary}

\section{Tradeoff II: Faster Preprocessing, More Space}\label{sec:oldtradeoff}

We now proceed to describe the tradeoff that was the main result in a preliminary version of this work~\cite{sodaversion}, and is encapsulated in the following theorem.\footnote{This result was obtained before the recent breakthroughs in exact distance oracles, which have now allowed us to get the tradeoffs of~\cref{cor:main}.}

\begin{restatable}{theorem}{thmmainsecond}\label{thm:main2}
For any integer $r \in [1,n]$ and for any integer $k \leq \frac{n}{r}$, there exists a data structure of size $\cO(\frac{n^{k+1}}{r^{k+1}} \sqrt{nr} + n \log^2 n)$, which can be constructed in time $\cOtilde(\frac{n^{k+1}}{r^{k+1}} \sqrt{nr} + n^2)$, and can answer the following queries in $\cO(k\sqrt{r} \log^2 n)$ time. Given vertices $u$ and $v$ and a set $X$ of at most $k$ failed vertices, report the length of a shortest $u$-to-$v$ path that avoids $X$.
\end{restatable}

For a fixed $r$, such that the oracles underlying both~\cref{cor:main,thm:main2} have query time roughly $k\sqrt{r}$, the oracles of~\cref{cor:main} require space smaller by roughly a factor of $n^{1/2}/r^{1/2}$ compared to the oracle that we present in this section.
Interestingly, however, the preprocessing time of the oracles of~\cref{cor:main} can be worse by polynomial factors for some range of values of~$r$. E.g., for $r=n^{1/4}$, compare $\frac{n^{k}}{r^{k}} \cdot n^{3/2}$ with $\frac{n^{k}}{r^{k}} \cdot n^{3/2}/r^{1/2}+n^2$.

\subsection{Voronoi Diagrams with Point Location}\label{sec:vor}
Let $P$ be a directed planar graph with real edge-lengths, and no negative-length cycles.
Let $S$ be a set of vertices that lie on a single face of $P$; we call the elements of $S$ sites. Each site $s\in S$ has a weight $\omega(s) \geq 0$ associated with it. The additively weighted distance between a site $s \in S$ and a vertex $v \in V$, denoted by $d^\omega_P(s,v)$ is defined as $\omega(s)$ plus the length of the $s$-to-$v$ shortest path in $P$.

\begin{definition}
The additively weighted Voronoi diagram of $(S,\omega)$ ($VD(S, \omega)$) within $P$ is a partition of $V(P)$ into pairwise disjoint sets, one set $\Vor(s)$ for each site $s \in S$. The set $\Vor(s)$ which is called the Voronoi cell of $s$, contains all vertices in $V(P)$ that are closer (w.r.t.~$d^\omega_P$(. , .)) to $s$ than to any other site in $S$ (assuming that the distances are unique). There is a dual representation $VD^{*}(S,\omega)$ of a Voronoi diagram $VD(S,\omega)$ as a planar graph with $\cO(|S|)$ vertices and edges.
\end{definition}

\begin{theorem}[\cite{Vorexact,Vordiam}]\label{thm:Vorexact}
Given subsets $S'_1,\ldots, S'_m$ of $S$, and additive weights $\omega_i(u)$ for each $u \in S'_i$, we can construct a
data structure of size  $\cO(|P|\log|P| + \sum_i|S'_i|)$ that supports the following (\emph{point location}) queries. Given $i$, and a vertex $v$ of $P$, report in $\cO(\log^2 |P|)$ time the site $s$ in the additively weighted Voronoi diagram $VD(S_i,\omega_i)$ such that $v$ belongs to $\Vor(s)$ and the distance $d^{\omega_i}_P(s,v)$.
The time and space required to construct this data structure are $\cOtilde(|P||S|^2+\sum_i|S'_i|)$.
\end{theorem}

\begin{remark}
Part of Theorem~\ref{thm:Vorexact} is proved in~\cite{Vorexact}, though not stated there explicitly as a theorem.
It is a tradeoff to Theorem~1.1 of~\cite{Vorexact}, requiring less space, and hence more applicable to our problem.
\end{remark}

\subsection{Handling a Single Failure}\label{sec:1f}
For ease of presentation we first describe an oracle that can handle just a single failure. 
We prove the following lemma, which is a restricted version of Theorem~\ref{thm:main2}.

\begin{lemma}
For any $r \in [1,n]$, there exists a data structure of size $\cO(\frac{n^{5/2}}{r^{3/2}} + n\log^2 n)$, which can be constructed in time $\cOtilde(\frac{n^{5/2}}{r^{3/2}} + n^2)$, and can answer the following queries in $\cO(\sqrt{r} \log^2 n)$ time. Given vertices $u,v,x$, report the length of a shortest $u$-to-$v$ path that avoids $x$.
\end{lemma}

\paragraph{Strategy.}
We change part II of the query, i.e.~computing $Y$-to-$v$ distances. After having computed $u$-to-$Y$ distances in $G \setminus X$, we identify an appropriate piece $Q$ in $\TG$ that contains $v$, and does not contain $u$ nor $x$. Exploiting the fact that distances within $Q$ remain unchanged when $x$ fails, we employ Voronoi Diagrams with point location for the piece $Q$, adapting ideas from~\cite{Vorexact}.

\paragraph{Additional preprocessing.}
For each pair of pieces $(R_i,R_j)$ of the $r$-division we compute and store the following. Let $S$ be a separator in the recursive decomposition, separating a piece into two subpieces $Q$ and $R$, such that $R_i \subseteq R$ and $R_j  \not\subseteq Q$.
For each $y \in \partial R_i \cup \partial R_j$,
for each hole $h$ of $Q$,
we compute and store a Voronoi diagram with the point location data structure for $Q$,
with sites the boundary vertices of~$Q$ that lie on $h$,
and additive weights the distances from $y$ to these sites in $G\setminus ((R_i\cup R_j)\setminus \{y\})$.

We now show that the space required is $\cOtilde(\frac{n^{5/2}}{r^{3/2}})$.
The space required for the first part of the query is $\cOtilde(\frac{n^2}{r})$ by~\cref{lem:utob}. We next analyze the space required for storing the Voronoi diagrams.
We consider $\cO(g(n,r,2)) = \cO(\frac{n^2}{r^2})$ pairs of pieces $(R_i,R_j)$, and for each of the $\cO(\sqrt{r})$ boundary vertices of each such pair we store, in the worst case, a Voronoi diagram for each of the $\cO(1)$ holes of each sibling of the nodes in the root-to-$R_i$ and root-to-$R_j$ paths in $\TG$.
The total number of sites of all Voronoi diagrams we store for a pair of pieces can be upper bounded by $\cO(\sqrt{n})$ by noting that the number of sites of a Voronoi diagram for a piece at level $\ell$ of $T_G$ is $\cO(\sqrt{n}/c_2^\ell)$ by~\cref{lem:rdiv}.
By Theorem~\ref{thm:Vorexact}, the space required to store a representation of a set of Voronoi diagrams with the functionality allowing for efficient point location queries for a piece $P$, with sites a subset of the boundary vertices of $P$, lying on a hole $h$ is $\cO\big(\sum_{P \in \TG} (\mathcal{S}_{P,h}+|P|\log|P|)\big)$, where $\mathcal{S}_{P,h}$ is the total cardinality of these sets of sites. Summing over all holes of all pieces $P$, noting that $\sum_{P \in \TG}\sum_h \mathcal{S}_{P,h} = \cO(\frac{n^{5/2}}{r^{3/2}})$ by the above discussion, and using~\cref{prop:1}, the total space required for all Voronoi diagrams is $\cO(\frac{n^{5/2}}{r^{3/2}}+n\log^2 n)$.

As for the preprocessing time, the precomputations of~\cref{lem:utob} take $\cOtilde(\frac{n^2}{r})$.
The additive weights can be computed in time $\cO(\frac{n^2}{r} \sqrt{nr} \log^3{n})$; see~\cref{lem:prec:vd} in~\cref{sec:prec}.
Further, we show in~\cref{lem:prec:vd2} that we can compute all required Voronoi diagrams in time $\cOtilde(n^2+\mathcal{S})$, where $\mathcal{S}$ is the size of their representation described in~\cref{sec:vor}.

\paragraph{Query.}
We first retrieve a piece $R_v$ of the $r$-division, containing $v$. 
We then proceed as follows (inspect~\cref{fig:query} for an illustration).

\begin{figure*}[t]
\resizebox{1.05\textwidth}{!}{
\subfloat[\large{Root-to-$R_i$ paths in~$\TG$.}]{
\begin{tikzpicture}

\node[shape=circle,scale=0.5,draw=black, fill, label = {right: $root$}] (r) at (0,0) {};
\node[shape=circle,scale=0.5,draw=black, fill] (1) at (1,-1) {};
\node[shape=circle,scale=0.5,draw=black, fill, label = {right: $S$}] (S) at (2,-2) {};
\node[shape=circle,scale=0.5,draw=black, fill, label = {left: $R$}] (R) at (1.5,-3) {};
\node[shape=circle,scale=0.5,draw=black, fill, label = {right: $Q$}] (Q) at (2.5,-3) {};

\node[shape=circle,scale=0.5,draw=black, fill, label = {left: $R_x$}] (Rx) at (0.2,-4) {};
\node[shape=circle,scale=0.5,draw=black, fill, label = {right: $R_u$}] (Ru) at (1.2,-4) {};
\node[shape=circle,scale=0.5,draw=black, fill, label = {right: $R_v$}] (Rv) at (3,-4) {};

\draw [->,decorate,decoration={snake,amplitude=.4mm,segment length=2mm,post length=1mm}] (r) -- (1);

\draw [->,decorate,decoration={snake,amplitude=.4mm,segment length=2mm,post length=1mm}] (1) -- (S);
\draw [->,decorate,decoration={snake,amplitude=.4mm,segment length=2mm,post length=1mm}] (1) -- (Rx);

\draw [->,decorate] (S) -- (R);
\draw [->,decorate] (S) -- (Q);

\draw [->,decorate,decoration={snake,amplitude=.4mm,segment length=2mm,post length=1mm}] (R) -- (Ru);

\draw [->,decorate,decoration={snake,amplitude=.4mm,segment length=2mm,post length=1mm}] (Q) -- (Rv);

\node[shape=circle,scale=0.5,draw=white] (align) at (3,-5) {};

\end{tikzpicture}

}
\subfloat[\large{The $u$-to-$v$ path in $G \setminus \{x\}$.}]{
{\includegraphics[width=\textwidth]{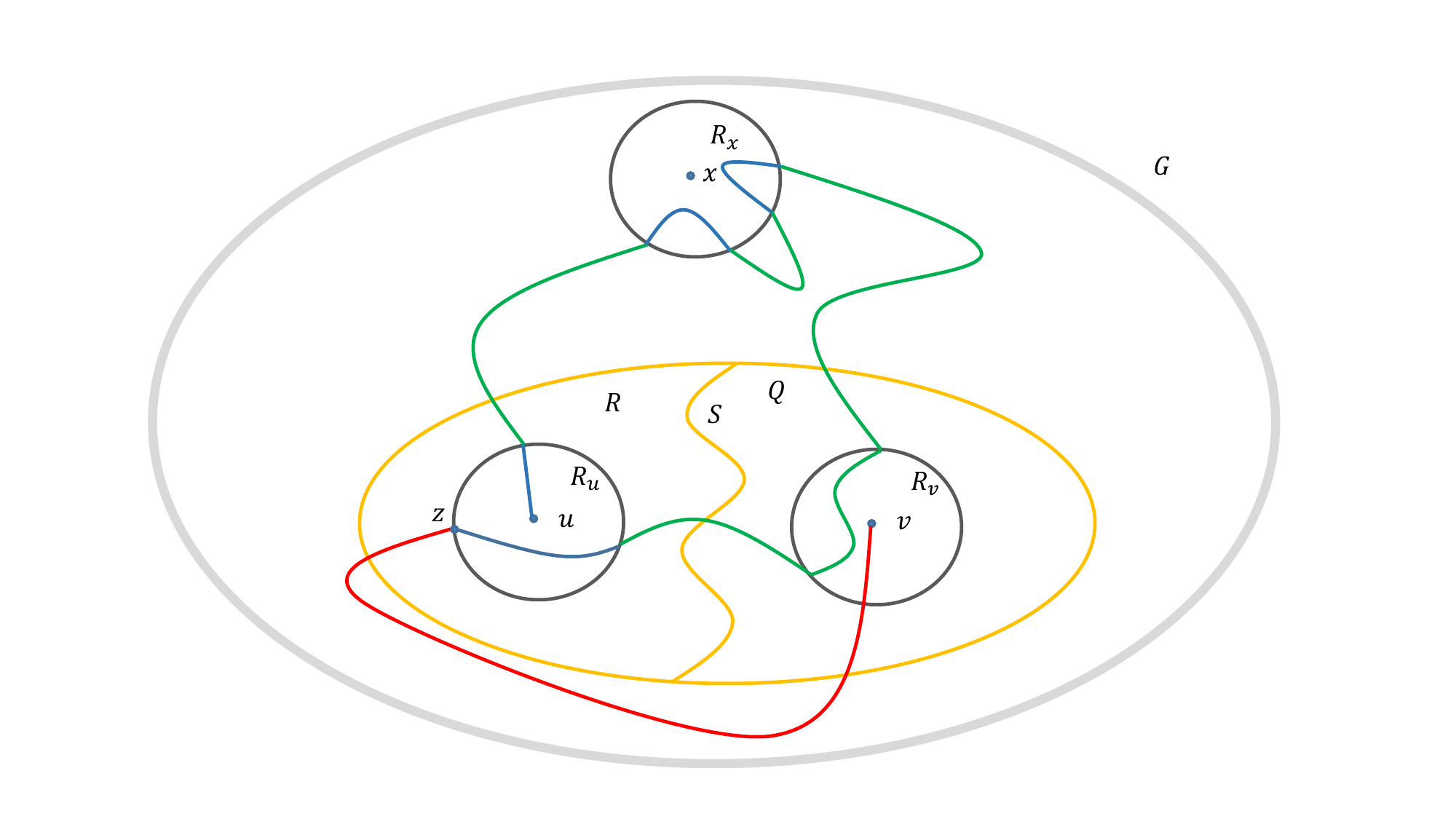}}
}
}
\caption{To the left: A view of the root-to-$R_i$ paths in $\TG$. Straight edges denote edges of the tree, while snake-shaped edges denote paths. To the right: A view of the shortest path in $G$. The paths in blue are represented by the $DDG^\circ$s, the paths in green by $DDG^\circ_{ext}$ and the length of the path in red is returned by the point location query in the Voronoi diagram.
}
\label{fig:query}
\end{figure*}

\begin{enumerate}
\item Following parent pointers from $R_v$ in $\TG$, we find the highest ancestor $Q$ of $R_v$ containing neither $u$ nor $x$.
Thus, the sibling $R$ of $Q$ in $\TG$ contains a vertex $i \in \{u,x\}$.
We find a descendant $R_i$ of $R$ that is in the $r$-division and contains $i$.
We then find any piece $R_j$ of the $r$-division containing the element of $\{u,x\}\setminus \{i\}$.
Note that, by choice of $Q$, $R_j$ is not a descendant of $Q$. Finding these pieces requires time $\cO(\log^2 n)$.
\item We perform a \textsc{$u$-to-$1$-Boundary}$(u,\{x\},(R_u,R_v),\textsf{null})$ query.
This takes time $\cO(\sqrt{r} \log^2 n)$ and returns $d_G(u,y,x)$ for each $y \in \partial R_u  \cup \partial R_x$.
\item For each $y \in (\partial R_u \cup \partial R_x)  \setminus \{x\}$, for each hole $h$ of $Q$,   
we perform an $\cO(\log^2 n)$-time query to the Voronoi diagram stored for $R_u, R_x, y$, and $h$ to get the distance from $y$ to $v$ in $G\setminus((R_u\cup R_x)\setminus\{y\})$.
The required distance is the minimum $d_G(u,y,x)+d(y,v,(R_u\cup R_x)\setminus\{y\})$ over all $y$.
Each query takes $\cO(\log^2 n)$ time and hence the total time required is $\cO(\sqrt{r} \log^2 n)$.
\end{enumerate}

We now argue the correctness of the query algorithm. Let $\rho$ be a shortest $u$-to-$v$ path that avoids $x$.
Let $z$ be the last vertex of $\rho$ that belongs to $\partial R_u \cup \partial R_x$. Let $h'$ be the hole of $Q$ such that the last vertex of $\rho$ that belongs to the boundary of $Q$ belongs to hole $h'$. The distance $d_G(u,z,x)$ from $u$ to $z$ in $G\setminus \{x\}$ is computed by the FR-Dijkstra computation in step 2, while the distance from $z$ to $v$ in $G\setminus \{x\}$ is obtained from the query to the Voronoi diagram stored for $R_u, R_x, z$, and $h'$.
It is easy to see that we do not obtain any distance that does not correspond to an actual path in $G\setminus \{x\}$ and hence the correctness of the query algorithm follows.

\subsection{Handling Multiple Failures}

We now explain how to straightforwardly generalize the approach presented in the previous subsections to obtain oracles that can handle multiple failures.

\paragraph{Preprocessing.}
\begin{enumerate}
\item We perform the precomputations of~\cref{lem:utob}.
\item For each $(k+1)$-tuple of pieces $(R_{i_1},\ldots,R_{i_{k+1}})$ of the $r$-division we compute and store the following. Let $S$ be a separator in the recursive decomposition, separating a piece into $Q$ and $R$, such that for some $j$, $R_{i_j}\subseteq R$ and none of the other pieces of the tuple is a subgraph of $Q$.
For each $y \in \bigcup\limits_{j=1}^{k+1}\partial R_{i_j}$,
for each hole $h$ of $Q$,
we store a Voronoi diagram with the point location data structure for $Q$,
with sites the boundary vertices of $Q$ that lie on $h$,
and additive weights the distances from $y$ to these sites in $G\setminus \big(\big(\bigcup\limits_{j=1}^{k+1} R_{i_j}\big)\setminus \{y\}\big)$.
\end{enumerate}

\paragraph{Query.} 

The algorithm is then essentially the same as that of~\cref{sec:1f}.
\begin{enumerate}
\item We find the highest ancestor $Q$ of $R_v$ in $\TG$ that does not contain any of the elements of $\{u\} \cup X$ and retrieve a descendant of its sibling in the $r$-division that does contain some element $i \in \{u\} \cup X$. We then identify a piece $R_j$ in the $r$-division for each $j \in \{u\} \cup X \setminus \{i\}$. This requires time $\cO(k \log^2 n)$.
\item We perform a \textsc{$u$-to-$k$-Boundary}$(u,X,\mathcal{R},v)$ query, for $\mathcal{R}=(R_u,R_{x_1}, \ldots , R_{x_k})$, which requires time $\cO(k\sqrt{r})$.
\item We perform $\cO(k\sqrt{r})$ point location queries to Voronoi diagrams of $Q$, each requiring time $\cO(\log^2 n)$.
\end{enumerate}

We hence obtain our second tradeoff theorem, restated here for convenience.

\thmmainsecond*

\begin{proof}
The correctness of the query algorithm follows by an argument identical to the one for the case of single failures (see~\cref{sec:1f}); its time complexity is analyzed above. We next analyze the space required by our data structure and its construction time.

\emph{Space Complexity.}
The space occupied by the data structure of~\cref{lem:utob} is $\cO(\frac{n^{k+1}}{r^{k}} \log n)$.
We bound the space required for the Voronoi diagrams by $\cO(g(n,r,k+1) k \sqrt{nkr}+n\log^2 n)$ as follows.
For each of the $\cO(k\sqrt{r})$ boundary vertices of each of the $\cO(g(n,r,k+1))$ $(k+1)$-tuples,
we store a Voronoi diagram for each of the $\cO(1)$ holes, of (at most) each of the siblings of the nodes in the root-to-$R_i$ path in $\TG$ for each $R_i$ in the tuple.
With an argument identical to the one used in the proof of Lemma~\ref{lem:fr}, the total number of boundary vertices (with multiplicities) of all of these pieces is $\cO(\sqrt{kn})$.
Hence the total number of sites of all Voronoi diagrams that we store is $\cO(g(n,r,k+1)k\sqrt{nkr})$.
By Theorem~\ref{thm:Vorexact}, the size required to store them with the required functionality is thus $\cO(g(n,r,k+1)k\sqrt{nkr} + \sum_{P \in \TG} |P| \log|P|)=\cO\left(\frac{(cn)^{k+1}}{r^{k+1}} \cdot \frac{1}{k!} \cdot \sqrt{nkr} + n \log^2 n\right)$, where the last equality follows by~\cref{prop:1}. 

Thus, since $k\leq n/r$, the total space required is \[\cO\left(\frac{(cn)^{k+1}}{r^{k+1}} \cdot \frac{1}{k!} \cdot (kr +  \sqrt{nkr}) + n \log^2 n\right) = \cO\left(\frac{(cn)^{k+1}}{r^{k+1}} \cdot \frac{1}{k!} \cdot \sqrt{nkr} + n \log^2 n\right).\]

\emph{Preprocessing time.} 
The preprocessing of~\cref{lem:utob} takes $\cO(\frac{n^{k+1}}{r^{k}} \log^2 n)$ time.
We can compute the required additive weights of all $(k+1)$-tuples in time $\cOtilde\left(\frac{(cn)^{k+1}}{r^{k+1}} \cdot \frac{1}{(k-1)!} \cdot \sqrt{nkr}\right)$, employing~\cref{lem:prec:vd}.
Finally, constructing the Voronoi diagrams requires time $\cOtilde(n^2+\mathcal{S})$, where $\mathcal{S}$ is the total size of their representation, which is equal to the total number of sites in these diagrams (with multiplicities), as shown in~\cref{lem:prec:vd2}; this dominates the time complexity.
\end{proof}

\section{Efficient Preprocessing}\label{sec:prec}

In this section we show how to efficiently compute the data structures described in~\cref{sec:newtradeoff,sec:oldtradeoff}.
Throughout this section, and similarly to Section~\ref{sec:FR}, when using FR-Dijkstra to compute $DDG^\circ$s, or other distances corresponding to shortest paths with a restriction on the vertices they can go through, we do not relax edges whose tail is a vertex that is not allowed to be on a shortest path.

It is shown in~\cite[Theorem~3]{DBLP:conf/stoc/KleinMS13} that, given a geometrically increasing sequence of numbers $\mathcal{r}=(r_1, r_2, \ldots, r_\nu)$, where $r_1$ is a sufficiently large constant, $r_{i+1}/r_i=b$, for all $i$, for some constant $b>1$, and $r_\nu=n$, we can obtain $r$-divisions for all $r\in \mathcal{r}$ in time $\cO(n)$ in total.
These $r$-divisions satisfy the property that a piece in the $r_i$-division is a weak descendant (in $\TG$) of a piece in the $r_j$-division for each $j>i$.

We first show how to efficiently compute the external $DDG^\circ$s for all $k$-tuples of pieces of an $r$-division, $r \in \mathcal{r}$.
Our algorithm is a natural adaptation of the top-down technique of~\cite{DBLP:journals/talg/BorradaileSW15}
for computing external DDGs to computing \emph{strictly external} DDGs of $k$-tuples.

\begin{lemma}\label{lem:prec:ddgext}
Given $r_i \in \mathcal{r}$ and an integer $d \leq \frac{n}{r_i}$, one can compute $DDG^\circ_{ext}$ for all $d$-tuples of pieces of each $r_t$-division, $t \geq i$, in time $\cO(\frac{(cn)^d}{r_i^{d-1}} \frac{1}{(d-2)!}  \log^2{n})$ for some constant $c>1$.
\end{lemma}
\begin{proof}
We prove this lemma by induction on $\mathcal{r}$ from top to bottom.
For $r_{\nu}=n$, the only piece is $G$, and  $DDG^\circ_{ext}(G)$ is the empty graph.
Assume inductively that we have $DDG^\circ_{ext}(R_1,\ldots,R_d)$ for every $d$-tuple $(R_1,\ldots,R_d)$ of pieces at the $r_{i+1}$-division.
Let $Q_1,\ldots,Q_d$ be pieces at the $r_i$-division. Note that every piece at level $r_i$ is contained in some piece at level $r_{i+1}$, but a piece at level $r_{i+1}$ might contain multiple pieces at level $r_i$. Let $R_1,\ldots,R_d$ be pieces of the $r_{i+1}$-division such that each $Q_j$ is a subgraph of some $R_{j}$; see~\cref{fig:topdown} for an illustration. (If the $r_{i+1}$-division has less than $d$ pieces we just take all of them.)
Let $\mathcal{Q}_{R_j}$ be the maximal subset of $\{Q_1,\ldots,Q_d\}$ such that each piece in $\mathcal{Q}_{R_j}$ is contained in $R_j$.
For every  $j\in \{1, \ldots, d\}$ let us denote the allowed internal part of $R_j$ by $R_j'$. Formally,
\[R_j'= R_j \setminus \left(\bigcup_{Q \in \mathcal{Q}_{R_j}} Q \setminus \partial Q\right).\]
Let us define the boundary of $R_j'$ to be 
\[\partial R_j \bigcup \left(\bigcup_{Q \in \mathcal{Q}_{R_j}} \partial Q\right).\]
Since $R_j$ and each $Q_m \in \mathcal{Q}_{R_j}$ have $\cO(\sqrt{r_{i+1}})$ and $\cO(\sqrt{r_i})$ boundary vertices respectively, $R_j'$ has $\cO(\sqrt{r_{i+1}} + \sqrt{r_i} |\mathcal{Q}_{R_j}|) = \cO(|\mathcal{Q}_{R_j}|\sqrt{r_{i+1}})$ boundary vertices (recall that $r_{i+1} / r_i = b$).

\begin{figure}[htpb!]
    \centering
    \includegraphics[width=10cm]{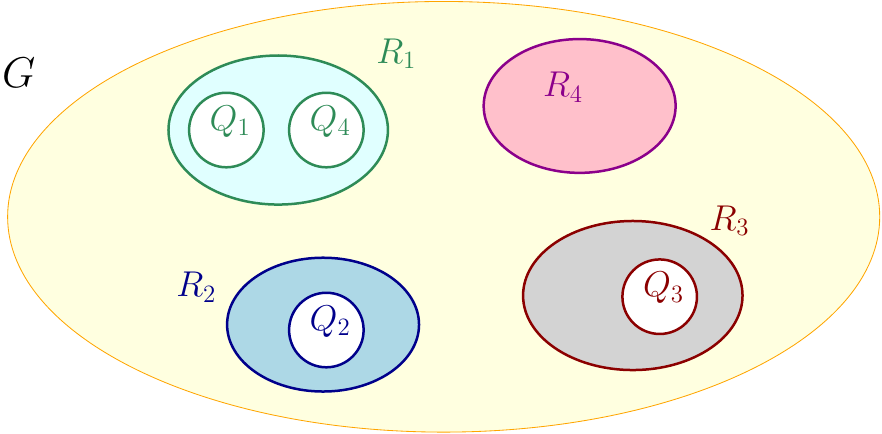}
    \caption{The setting in the proof of~\cref{lem:prec:ddgext}. 
    We have $\mathcal{Q}_{R_1}=\{Q_1,Q_4\}$ and $\mathcal{Q}_{R_4}=\emptyset$. 
    For each~$j$, $R_j'$ is the colored part of $R_j$. 
    For instance, $R_2'=R_2 \setminus (Q_2 \setminus \partial Q_2)$ and $R_4'=R_4$.}
    \label{fig:topdown}
\end{figure}

Let $DDG^\circ_{R_j'}$ be the a complete directed graph on the boundary vertices of $R_j'$ such that the edge $(u,v)$ has weight $d^\circ_{R_j'}(u,v)$ equal to the length of the shortest $u$-to-$v$ path in $R_j'$ that is internally disjoint from the boundary of $R_j'$. 

We compute $DDG^\circ_{R_j'}$ in a similar manner to  the query of \cref{sec:FR} by running FR-Dijkstra on the union of the following $DDG^\circ$s.
For each piece $Q_m \in \mathcal{Q}_{R_j}$, for each ancestor $Q$ of $Q_m$ (including $Q_m$) that is a strict descendant of $R_j$ in $\TG$, we take the $DDG^\circ_P$ of the sibling $P$ of $Q$ if $P$ contains no piece of $\mathcal{Q}_{R_j}$.
The pieces of $\mathcal{Q}_{R_j}$ have $\cO(|\mathcal{Q}_{R_j}| \sqrt{r_i})$ boundary vertices in total and the total number of boundary vertices for their considered ancestors is bounded by $\cO(|\mathcal{Q}_{R_j}| \sqrt{r_{i+1}})$, as the number of boundary vertices in any root-to-leaf path in $\TG$ decreases geometrically (cf.~\cref{lem:rdiv}).
Running FR-Dijkstra from each of the $\cO(|\mathcal{Q}_{R_j}|\sqrt{r_{i+1}})$ boundary vertices of $R_j'$ yields $DDG^\circ_{R_j'}$  and requires $\cO(|\mathcal{Q}_{R_j}|\sqrt{r_{i+1}} |\mathcal{Q}_{R_j}| \sqrt{r_{i+1}} \log^2{n})=\cO(|\mathcal{Q}_{R_j}|^2r_{i+1}\log^2{n})$ time in total. 
When summing over $R_1, \ldots, R_d$ we get 
\[\sum_{j=1}^{d}{|\mathcal{Q}_{R_j}|^2  \cdot r_{i+1}  \cdot \log^2{n}} 
\leq r_{i+1} \cdot \log^2{n} \cdot \left(\sum_{j=1}^{d}{|\mathcal{Q}_{R_j}|}\right)^2 
= d^2 \cdot r_{i+1} \cdot \log^2{n}.\] 
Note that the equality follows from the fact that $\sum_{j=1}^{d}{|\mathcal{Q}_{R_j}|}=d$. 

Let $\mathcal{D} = DDG^\circ_{ext}(R_1,\ldots,R_d)\bigcup (\bigcup _{j=1}^{d} DDG^\circ_{R_j'})$.
Each of $DDG^\circ_{ext}(R_1,\ldots,R_d)$ and $\bigcup_{j=1}^{d} DDG^\circ_{R_j'}$ contributes $\cO(d\sqrt{r_{i+1}})$ boundary vertices to $\mathcal{D}$.
We run FR-Dijkstra on $\mathcal{D}$ from each boundary vertex of $Q_m$ for $m\in \{1, \ldots d\}$ to obtain $DDG^\circ_{ext}(Q_1,\ldots,Q_d)$. 
There are $\cO(d\sqrt{r_i})$ such boundary vertices, so this requires time $\cO(d\sqrt{r_i} d (\sqrt{r_{i+1}} + \sqrt{r_i}) \log^2{n})=\cO(d^2 \cdot r_{i+1} \cdot \log^2{n})$.

We can thus compute $DDG^\circ_{ext}(Q_1, \ldots, Q_d)$ for all $d$-tuples at level $r_i$ in time \[\cO((g(n,r_i,d) \cdot d^2 \cdot r_{i+1} \cdot \log^2{n}) = \cO\left(\frac{(cn)^d}{r_i^{d}} \cdot r_{i+1} \cdot \frac{1}{d!} \cdot d^2 \cdot \log^2{n}\right) = \cO\left(\frac{(cn)^d}{r_i^{d-1}} \cdot \frac{1}{(d-2)!} \cdot \log^2{n}\right),\] assuming that we have the $DDG^\circ_{ext}$s for all $d$-tuples of pieces of $r_t$-divisions, $t>i$.

The time to compute the $DDG^\circ_{ext}$s for all $d$-tuples of pieces of all $r_t$-divisions, $t>i$, is, inductively, \[\cO\left((cn)^d \cdot \frac{1}{(d-2)!} \cdot \log^2{n} \cdot\sum_{t=i+1}^{\nu}\frac{1}{r_t^{d-1}}\right), \text{ and } \sum_{t=i+1}^{\nu}\frac{1}{r_t^{d-1}}=\frac{1}{r_i^{d-1}}\sum_{t=1}^{\nu-i}\left(\frac{1}{b^{d-1}}\right)^t=\cO\left(\frac{1}{r_i^{d-1}}\right)\] since $b^{d-1}>1$.
Thus, computing the $DDG^\circ_{ext}$s for $d$-tuples of pieces of the $r_i$-division dominates the time complexity.
\end{proof}

We next show how to efficiently compute the additive distances with respect to which the Voronoi diagrams stored by our oracle are computed. 

\begin{lemma}\label{lem:prec:vd}
Let $\mathcal R_r$ be an $r$-division, such that $r \in \mathcal{r}$, and let $d \leq \frac{n}{r}$ be an integer. For all $d$-tuples of pieces $R_1, \ldots, R_d$ in $\mathcal R_r$ and for all pieces $Q \in \TG$ such that $Q$ does not contain any of the pieces $R_i$, and $Q$ is a sibling of a node in the root to-$R_i$ path in $\TG$ for some $R_i$, one can compute the distances from each $y \in \bigcup\limits_{i=1}^{d} \partial R_{i}$ to each boundary vertex of $Q$ in the graph $G\setminus \big(\big(\bigcup\limits_{i=1}^{d} R_{i}\big)\setminus \{y\}\big)$ in time $\cO\big(\frac{(cn)^d}{r^d} \cdot \frac{1}{(d-2)!} \cdot \sqrt{ndr} \cdot \log^3{n}\big)$ in total, for some constant $c>1$.
\end{lemma}
\begin{proof}
Let us consider a $d$-tuple of pieces $(R_1, \ldots, R_d)$ and a piece $Q$, satisfying the properties in the statement of the lemma.
To compute the desired distances, we run FR-Dijkstra from each $y \in \bigcup_{i=1}^{d} \partial R_{i}$ on the union of the following $DDG$s:
\begin{enumerate}
\item $DDG^\circ_{Q}$.
\item For each piece $R_i \in \{R_1, \ldots, R_d\}$  for each ancestor $A$ of $R_i$ (including $R_i$) in $\TG$, we take the $DDG^\circ_{B}$ of the sibling $B$ of $A$ if $B$ contains no piece of $R_1, \ldots, R_d$. 
\end{enumerate}

This correctly computes the distances by the same arguments that were applied in \cref{sec:FR}.
It remains to analyze the time complexity. Consider the $(n/d)$-division of $G$ in $\TG$. By the same argument that was applied in the proof of~\cref{lem:fr} we can bound the number of boundary vertices for all the included $DDG^\circ$s by $\cO(\sqrt{dn})$.
There are $\cO(d\sqrt{r})$ choices of $y \in \bigcup\limits_{i=1}^{d} \partial R_{i}$, so the time required to run FR-Dijkstra from each $y$ 
is $\cO(d\sqrt r  \cdot \sqrt{dn} \cdot \log^2 n) = \cO(d \cdot \sqrt{nrd} \cdot \log^2{n})$.

Each piece $R_i \in \{R_1, \ldots , R_d \}$ has $\cO(\log{n})$ nodes in the root-to-$R_i$ path in $\TG$, hence computing the distances for all possible choices of $Q$
requires time $\cO(d^2 \sqrt{nrd} \log^3{n})$. 
Finally, in order to compute the distances for all $d$-tuples of pieces we need time \[\cO((g(n,r,d) \cdot d ^2 \cdot \sqrt{nrd} \cdot \log^3{n}) = \cO\left(\frac{(cn)^d}{r^d} \cdot \frac{1}{d!} \cdot d^2 \sqrt{nrd} \cdot \log^3{n}\right), \text{ as claimed}.\qedhere\]
\end{proof}

\begin{lemma}\label{lem:prec:vd2}
We can compute the representation of the Voronoi diagrams described in~\cref{sec:prel} with respect to sets of sites of total cardinality $\mathcal{S}$, each corresponding to a piece $P\in \TG$ and  consisting of nodes of $\partial P$ that lie on a single hole of $P$, and specifying an additive weight for each of these nodes in time $\cOtilde(n^2+\mathcal{S})$ in total.
\end{lemma}
\begin{proof}
We apply Theorem~\ref{thm:Vorexact} and construct all the Voronoi diagrams corresponding to each of the $\cO(1)$ holes of each piece as a batch.
For a hole $h$ of a piece $P$, the time required is $\cOtilde(|P||\partial P|^2+\sum_h\mathcal{S}_{P,h})$, where $\mathcal{S}_{P,h}$ is the total cardinality of the sets of sites corresponding to nodes of $\partial P$ lying on $h$.
Then we have that
$$\sum_{P \in \TG} \big(|P||\partial P|^2+\sum_h|\mathcal{S}_{P,h}|\big) = \cO(n^2+\mathcal{S}),$$
by~\cref{prop:2} and hence the stated bound follows.
\end{proof}

\section{Final Remarks}
Perhaps the most intriguing open question related to our results is whether it is possible to answer distance queries subject to even one failure in time $\cOtilde(1)$ with an $o(n^2)$-size oracle.

\bibliographystyle{plainurl}

\end{document}